\documentclass[preprint,12pt,3p]{elsarticle}
\usepackage{microtype}
\usepackage[T1]{fontenc}
\usepackage[utf8]{inputenc}
\usepackage{amsmath, amsthm, amssymb, mathtools}

\usepackage{enumitem, paralist}
  \pltopsep=\medskipamount
\usepackage{booktabs}
\usepackage[flushleft]{threeparttable}

\usepackage{subcaption}
\usepackage{pdflscape}
\usepackage{multirow}
\usepackage{url}
\usepackage{xfrac}
\usepackage{siunitx}
\DeclareSIUnit{\hour}{h}
\DeclareSIUnit{\kWh}{kWh}
\usepackage{thmtools,thm-restate}

\newtheorem{prop}{Proposition}

\usepackage{hyperref}
\usepackage{tikz}
\usetikzlibrary{graphs, quotes, decorations.pathreplacing,positioning, arrows.meta}

\usepackage{pgf, pgfplots}
\usepgfplotslibrary{groupplots}
\pgfplotsset{every boxplot/.style={mark=*,every mark/.append style={mark size=1pt}}}

\makeatletter
\pgfplotsset{
    groupplot xlabel/.initial={},
    every groupplot x label/.style={
        at={($({\pgfplots@group@name\space c1r\pgfplots@group@rows.west}|-{\pgfplots@group@name\space c1r\pgfplots@group@rows.outer south})!0.5!({\pgfplots@group@name\space c\pgfplots@group@columns r\pgfplots@group@rows.east}|-{\pgfplots@group@name\space c\pgfplots@group@columns r\pgfplots@group@rows.outer south})$)},
        anchor=north,
    },
    groupplot ylabel/.initial={},
    every groupplot y label/.style={
            rotate=90,
        at={($({\pgfplots@group@name\space c1r1.north}-|{\pgfplots@group@name\space c1r1.outer
west})!0.5!({\pgfplots@group@name\space c1r\pgfplots@group@rows.south}-|{\pgfplots@group@name\space c1r\pgfplots@group@rows.outer west})$)},
        anchor=south
    },
    execute at end groupplot/.code={%
      \node [/pgfplots/every groupplot x label]
{\pgfkeysvalueof{/pgfplots/groupplot xlabel}};  
      \node [/pgfplots/every groupplot y label] 
{\pgfkeysvalueof{/pgfplots/groupplot ylabel}};  
    }
}

\def\endpgfplots@environment@groupplot{%
    \endpgfplots@environment@opt%
    \pgfkeys{/pgfplots/execute at end groupplot}%
    \endgroup%
}
\makeatother

\usepackage[numbers]{natbib}

\begin{document}
\begin{frontmatter}
\title{The vehicle routing problem with drones and drone speed selection}
\author[1]{Felix Tamke\corref{cor1}}
\ead{felix.tamke@tu-dresden.de}
\author[1]{Udo Buscher}
\ead{udo.buscher@tu-dresden.de}

\address[1]{Faculty of Business and Economics, TU Dresden, Dresden 01062, Germany}

\cortext[cor1]{Corresponding author}

\begin{abstract}
Joint parcel delivery by trucks and drones has enjoyed significant attention for some time, as the advantages of one delivery method offset the disadvantages of the other. This paper focuses on the vehicle routing problem with drones and drone speed selection (VRPD-DSS), which considers speed-dependent energy consumption and drone-charging in detail. For this purpose, we formulate a comprehensive mixed-integer problem that aims to minimize the operational costs consisting of fuel consumption costs of the trucks, labor costs for the drivers, and energy costs of the drones. The speed at which a drone performs a flight must be selected from a discrete set. We introduce preprocessing steps to eliminate dominated speeds for a flight to reduce the problem size and use valid inequalities to accelerate the solution process.
The consideration of speed-dependent energy consumption leads to the fact that it is advisable to perform different flights at different speeds and not to consistently operate a drone at maximum speed. Instead, drone speed should be selected to balance drone range and speed of delivery. Our extensive computational study of a rural real-world setting shows that, by modeling energy consumption realistically, the savings in operational costs compared to truck-only delivery are significant but smaller than those identified in previously published work. Our analysis further reveals that the greatest savings stem from the fact that overall delivery time decreases compared to truck-only delivery, allowing costly truck-driver time to be reduced. The additional energy costs of the drone, however, are largely negligible.
\end{abstract}

\end{frontmatter}

\section{Introduction}
This paper introduces the vehicle routing problem with drones and drone speed selection (VRPD-DSS). The integration of drones into transportation systems for parcel delivery has attracted significant attention in recent years \cite{chungs-2020, macrinag-2020, moshrefm-2021, ottoa-2018, rojasv-2021}. In contrast to delivery trucks, drones are not restricted to road networks and are, therefore, considered faster. However, they have several drawbacks. Drones have very limited capacities, e.\,g., in terms of the number of parcels or maximum payload, as well as limited range. Therefore, they are not well-suited as a stand-alone solution for delivery scenarios involving longer distances. One way to offset the disadvantages of drones is to combine trucks and drones into tandems \cite{agatzn-2018, daimler-2019, murrayc-2015, workhorse-2015}. In these cases, a truck carries one or more drones atop its roof and operates as a kind of mobile warehouse for drones but it also makes deliveries to customers.

Truck-drone tandems are considered a viable option for last-mile delivery, especially in rural areas \cite{boysenn-2021, gabaf-2020}, where distances between customers are greater because the population density is lower. Thus, trucks have to travel long distances to make deliveries to remote customers, while drones can fly more-direct routes. However, direct drone delivery from a warehouse or store, as tested by various companies for urban and suburban areas, is usually not possible due to the long distances. Hence, the trucks are able to work as range extenders for the drones. Moreover, drones can usually operate in rural areas without interference from tall buildings or other obstacles. This allows for less-stringent regulations concerning their use. Another aspect to consider is a potential customer-delivery zone. Residences in rural areas are mostly stand-alone houses. This allows a drone to drop the package somewhere on the customer's property, for example with a winch, which simplifies the delivery process.

One of the most important factors of drone use is their limited range. In the literature on truck-drone tandems, the range is often only an approximation based on a maximum distance or a time limit \cite{macrinag-2020, zhangj-2021}. However, both approaches are simplifications for the maximum energy consumption of a drone. The energy consumed by a drone depends on several factors. \citet{zhangj-2021} grouped these factors with respect to drone design, environment, drone dynamics, and delivery operations. These authors showed that, in particular, the total weight of the drone at takeoff, which consists of the payload plus drone and battery weight, and the airspeed have a major impact on energy consumption and, thus, on range. For a given drone configuration and just a single customer delivery per flight, the total weight at takeoff is fixed, and only the speed of the flight can be adjusted to vary the range. 

However, in most optimization approaches for truck-drone tandems, the range is either independent of the speed, or the speed is the same for all flights. Thus, speed is not included in the decision-making process. In the context of the problem considered here, i.e., trucks and drones make deliveries, to the best of our knowledge, variable drone speeds affecting the range are presented only in \cite{rajr-2020} for a single truck with multiple drones. They introduced a heuristic approach in which drone speeds can be dynamically adjusted and demonstrated that significant time-savings can be achieved with variable drone speeds. 

In contrast to \cite{rajr-2020}, we present for the first time an exact approach for routing truck-drone tandems that takes into account speed-dependent energy consumption and other relevant aspects such as recharging. Here, considering drone speed as a continuous decision variable is not practical because the speed affects the energy expended in a nonlinear manner \cite{zhangj-2021}. Therefore, we perform a discretization of the speed to obtain different discrete levels. The energy consumption for each discrete level can be determined in advance and, thus, a speed has to be selected for a flight. We call this new problem the vehicle routing problem with drones and drone speed selection and make the following additional contributions in this paper:
\begin{compactitem}
    \item We conduct numerical experiments on realistic instances for a rural scenario. The instances incorporate real-world routing between locations and realistic drone parameters. The results show that substantial cost-savings can be achieved by truck-drone tandems in comparison to traditional truck-only delivery for the given rural test instances. However, because of the more-realistic assumptions, the savings are not as high as those shown in other publications. 
    \item We show in our experiments that the greatest savings can be achieved by shortening delivery times, while the power costs of the drones are almost negligible.
    \item We find that the speed selected in advance has a large impact on the costs of the VRPD when only a single speed is available. In contrast, the VRPD-DSS is independent of this pre-selected speed and achieves at least the minimal costs of all VRPDs but usually provides better solutions.
    \item We present and prove preprocessing methods to eliminate dominated speeds for drone flights and unnecessary variables to efficiently reduce the problem size without excluding optimal solutions.
\end{compactitem}

The paper is structured as follows. A brief overview of the related literature is provided in Section~\ref{sec:Literature}. We then present the assumptions made to define the VRPD-DSS and explain the energy-consumption model used in this paper in detail in Section~\ref{sec:assumptions}. A mixed-integer linear programming formulation for the VRPD-DSS is introduced in Section~\ref{sec:problemDefinition}. Section~\ref{sec:Preprocessing} presents the newly developed preprocessing methods. Known and new valid inequalities to strengthen the model formulation are described in Section~\ref{sec:valid_inequalities}. The generation of the rural test instances and the results of our computational studies are reported in Section~\ref{sec:ComputationalStudies}. Finally, we conclude the paper in Section~\ref{sec:Conclusion}.

\section{Related literature}\label{sec:Literature}
The literature on optimization problems associated with drones is growing rapidly as shown in recent surveys \cite{chungs-2020, macrinag-2020, moshrefm-2021, ottoa-2018, rojasv-2021}. Therefore, we focus our brief review of the relevant literature mostly on the drone range concepts used and on problems where trucks and drones can perform deliveries. \citet{macrinag-2020} distinguished this class of problems from problem classes where only drones are able to make deliveries, either from stationary depots (drone delivery problem) or from mobile ground vehicles (carrier-vehicle problem with drones). We refer to the most recent surveys in \citet{chungs-2020}, \citet{macrinag-2020}, and \citet{moshrefm-2021} for a more detailed review of these problems.

Combined delivery by trucks and drones as tandems was introduced by \citet{murrayc-2015} as the flying sidekick traveling salesman problem (FSTSP) for \textit{a single truck with a single drone}. They limited the range of the drone by a maximum time it can be airborne and presented a mixed-integer linear programming (MILP) model as well as a heuristic approach to solve the FSTSP. \citet{agatzn-2018} presented a variant where drones can perform loops, i.e., start and end a flight at the same node. They called this problem the traveling salesman problem with drones (TSPD) and introduced two heuristics. In contrast to \cite{murrayc-2015}, the drone range is limited by a maximum flight distance. In addition, they varied the speed of the drone in the computational experiments and showed that a faster drone leads to significantly reduced costs. However, the range is independent from the speed. Many additional heuristic algorithms, e.g., \cite{campuzanog-2021, defreitasj-2020, haq-2020, jeongh-2019, poikonens-2019}, and exact approaches, e.g., \cite{bocciam-2021, boumanp-2018, dellamicom-2021a, dellamico-2021b, robertir-2021, vasquezs-2021}, have been designed to solve these two basic problems or closely related variants. Currently, the best-known exact approach for a single tandem with one drone is a branch-and-price algorithm introduced in \cite{robertir-2021}. Unlike the other approaches, \citet{jeongh-2019} used an approximation of a simple energy-consumption function that takes into account the loaded weight. In a recent heuristic approach \cite{campuzanog-2021}, the authors used a more detailed energy-consumption model that considers the parcel weight and a fixed speed. 

A natural extension of the basic problems FSTSP and TSPD is to consider \textit{multiple drones for one truck}, as in \cite{cavanis-2021, dellamico-2021c, murrayc-2020, rajr-2020}, for example. \citet{murrayc-2020} introduced the multiple flying sidekicks traveling salesman problem (mFSTSP) and focused on the scheduling of launch and retrieval operations of multiple drones to deal with the small space on a delivery truck. They present an MILP model to solve small instances and a heuristic algorithm for larger instances. In addition, they investigated different approaches to drone endurance, including an energy-consumption model based on parcel weight and speed. However, the speed is not part of the decision-making process. In their computational experiments, they demonstrated that not using an actual energy-consumption model often leads to under-utilization of resources or infeasible solutions in terms of energy consumption. In a subsequent paper \cite{rajr-2020}, the authors included the drone speed as a decision variable and called the resulting problem the mFSTSP with variable drone speeds. They introduced a heuristic approach to solve the problem and showed that variable drone speeds lead to substantial time-savings. 

Another extension of the FSTSP and the TSPD arises from using \textit{multiple tandems with one or more drones per truck}. This problem is introduced in \citet{wangx-2017} as the vehicle routing problem with drones (VRPD). As for the problem with one tandem, several heuristic algorithms, e.g. \cite{daknamar-2017, dipuglia-2020, kitjacharoenchaip-2019, liuy-2020, sacramentod-2019, schermerd-2019a, wangd-2019}, and exact approaches, e.g. \cite{dipuglia-2017, schermerd-2019b, tamkef-2021, wangz-2019}, have been presented for the VRPD and several related problems. Several approaches consider time windows \cite{dipuglia-2017, dipuglia-2020}, allow multiple visits to customers on the same flight \cite{liuy-2020, wangz-2019}, or enable the launch and retrieval of drones at discrete points on an arc \cite{schermerd-2019a}. However, only the approach presented by \citet{liuy-2020} considers a range that is not limited by a maximum time or distance. Instead, they used an approximation of a drone's energy consumption that depends on the loaded weight.

Different concepts for a drone's range are also applied in \textit{problems without combined deliveries}, e.g., in \cite{chengc-2020, poikonens-2020, dukkancio-2021}. \citet{chengc-2020} developed an exact algorithm for a drone delivery problem with multiple trips and with non-linear energy consumption based on the payload. \citet{poikonens-2020} studied a problem with one truck and multiple drones and proposed a heuristic algorithm. Drones are allowed to visit multiple customers on the same flight, but the truck cannot visit a node when a drone is in the air. The energy consumption used in their approach takes into account the loaded weight of the parcels. Two different drone speeds are tested in their computational experiments, but the speed has no impact on the expended energy. \citet{dukkancio-2021} considered a problem where drones are first transported to launch points by trucks, serve a customer, and then return to the truck to start the next service. The trucks remain at the launch points and perform no deliveries, while the drones serve the customers. The authors determined the energy consumption explicitly and used the speed of the drones as continuous decision variables. To solve this problem, they reformulated the non-linear model into a second-order cone-programming problem.

In summary, to the best of our knowledge, there is currently no approach for the VRPD that takes into account the actual energy consumption of drones and incorporates drone speed into the decision-making process. Additionally, in contrast to battery-switching, recharging the battery while the drone is on the truck has not yet been included. Therefore, we adapt and extend one of the currently best exact approaches for the VRPD, as presented in \cite{tamkef-2021}, to address these relevant and important additional features.

\section{Preliminary considerations}\label{sec:assumptions}
\subsection{Assumptions on drone operations and truck-drone interaction}
To model the VRPD-DSS, we make the following assumptions:
\begin{compactenum}[(a)]
	\item A truck can be equipped with one or more drones. However, each drone is associated with one truck exclusively. Therefore, that drone may not be launched or received by any other truck. This is reasonable since the technological effort required to coordinate multiple drones on one truck is high.
	\item Trucks and drones do not have to use the same distance metric in a network because trucks are bound by the road network whereas drones are not.
	\item Each drone operation comprises three steps. First, it has to be launched from the truck. Next, the drone performs a delivery to exactly one customer, and then, it returns to its associated truck.
	\item A drone must be launched and retrieved at nodes of the given network, i.e., the depot or customer locations. A drone must not start and end a flight at the same customer location. In addition, a truck must not return to an already-visited customer to retrieve a drone. Likewise, trucks and drones may return to the depot only once.	
	\item We consider service times at customer locations for truck as well as drone deliveries. We also take into account the time needed to prepare a launch. The time required to retrieve a drone is not considered since we assume that the drones operate autonomously.
	\item A drone can fly at different speeds. The speed for a flight can be selected from a discrete set of available speeds and is constant during the whole flight (steady flight). The speed of the truck is given and is not part of the decision-making process.
	\item A drone expends energy by flying and hovering. Its energy consumption while flying depends on the selected speed and the weight. Hovering occurs in two cases: first, if the drone has to wait at the retrieval location and, second, while it is serving a customer. Other operations like climb and descent are not taken into consideration. We assume that the amount of time not spent hovering or in steady flight is negligible.
	\item The energy that can be expended during flight and hovering is limited by the available energy of the drone battery. However, the battery can be charged at a constant rate when the drone is on top of the truck. It cannot be recharged while the drone is being prepared for a launch.
\end{compactenum}

\subsection{Energy-consumption model} \label{sec:energy_use_model}
An energy-consumption model for drone operations is essential in our study. We use the model presented in \cite{stolaroffj-2018} for two reasons. First, multi-rotor drones are considered, corresponding to the technology currently used in truck-drone tandems. Second, the model provides different power functions for steady flight and hovering. Thus, we are able to distinguish between these two operations and can include the speed of a flight into energy considerations and decision-making. We use the octocopter presented in \cite{stolaroffj-2018} in all our tests (see Table~\ref{tab:parameters_drone} for parameters) because octocopters are capable of carrying heavier packages and, thereby, are suitable for truck-drone tandems.

A drone consumes energy when hovering and flying because it has to resist both gravity and drag forces. The latter are caused by the forward motion of the drone and the wind. However, we assume perfect ambient conditions such as no wind. All of a drone's activities in the air are accomplished by adjusting the speed of each rotor. This leads to the required thrust and pitch to perform an operation, e.g., moving forward at a desired speed. The thrust of the individual rotors differs depending on the type of activity. On average, however, they are almost equal and, together, exactly balance gravity and drag forces. Therefore, the total required thrust $T$ can be described as the sum
\begin{equation}
	T = F^{\text{g}} + F^{\text{d}}
\end{equation}
of gravity $F^{\text{g}}$ and the total drag forces $F^{\text{d}}$.

As described in \cite{stolaroffj-2018}, it is convenient to divide the drone into three relevant components: drone body (db), battery (b), and a customer's package (p). All three components $i \in DP = \left\{ \text{db}, \text{b}, \text{p} \right\}$ have the same attributes: mass $m_i$, drag coefficient $c^{\text{d}}_{i}$, and projected area $A_i$ perpendicular to the direction of travel. Hence, gravity $F^{\text{g}}$ is equal to
\begin{equation}
	 F^{\text{g}}  = g \sum_{i \in DP} m_i,
\end{equation}
where $g$ is the standard acceleration due to gravity. The total drag force $F^{\text{d}}$ for steady flight with speed $v$ can be estimated with the equation
\begin{equation}
	F^{\text{d}} = \frac{1}{2} \rho v^2 \sum_{i \in DP}   c^{\text{d}}_{i} A_i,
\end{equation}
where $\rho$ is the density of air. 

For a drone with $n$ rotors of diameter $D$, we are now able to calculate the theoretical minimum power required to hover as
\begin{equation}
	P^{\text{H,min}} = \frac{T^{\frac{3}{2}}}{\sqrt{\frac{1}{2} \pi n D^2 \rho}}.
\end{equation}
Note that for hovering, $v = 0$ and, therefore, $F^{\text{d}} = 0$ and $T = F^{\text{g}}$. The theoretical minimum power for steady flight with speed $v$ is given by
\begin{equation}
P^{\text{F,min}} = T\left( v \sin \alpha +v^{\text{i}} \right)
\end{equation}
with pitch angle $\alpha$ and induced speed $v^{\text{i}}$. Pitch angle $\alpha$ is the tilt of the drone in the direction of travel and can be determined by
\begin{equation}
	\alpha = \arctan \left( \frac{F^{\text{d}}}{F^\text{g}} \right).
\end{equation}
The induced speed at the rotors $v^{\text{i}}$ can be computed by solving the implicit equation
\begin{equation}
	v^{\text{i}} - \frac{2T}{\pi n D^2 \rho \sqrt{\left( v \cos \alpha \right)^2 + \left( v \sin \alpha + v^{\text{i}} \right)^2}} = 0,
\end{equation}
which can easily be done numerically. Finally, we consider an overall power efficiency of the drone $\eta$ and also use a safety coefficient $\sigma$. The latter is used to reflect circumstances not considered in the energy-consumption function, such as wind and temperature, and to prevent an underestimation of the power consumption. Hence, the expended power during hover $P^{\text{H}}$ and forward flight $P^{\text{F}}$ is expressed as
\begin{equation}
P^{\text{H}} = \frac{P^{\text{H,min}}}{\eta} \left(1 + \sigma\right) \text{ and }  P^{\text{F}} = \frac{P^{\text{F,min}}}{\eta} \left(1 + \sigma\right).
\end{equation}

In our studies, we vary two parameters that are relevant for computing the energy consumption of a given drone configuration: speed $v$ and the mass of the package for a customer $m_{\text{p}}$. Hence, we introduce the expended power for hovering and steady forward flight as functions $P^{\text{H}}(m_{\text{p}})$ and $P^{\text{F}}(m_{\text{p}},v)$. Note that the drag coefficient $c^d_{\text{p}}$ and the projected area $A^d_{\text{p}}$ also change with different packages due to their different shapes. However, these effects are negligible compared to those of other two factors. 

In the following, we analyze the trade-off between drone speed and energy consumed for the given model. Consider the examples shown in Figure~\ref{fig:powerAndEnergy} to better understand the relationship between package mass, speed, and corresponding expended power (\ref{fig:powerAndEnergyA}) and energy consumption (\ref{fig:powerAndEnergB}, \ref{fig:powerAndEnergyC}) for the introduced model. We distinguish between two different energy-consumption scenarios. First, we take into account only the energy consumption for flying \SI{1000}{\metre} (\ref{fig:powerAndEnergB}). Secondly, we assume that the drone flies \SI{1000}{\metre} and then has to hover and wait for the truck, which arrives \SI{180}{\second} after the drone takes off (\ref{fig:powerAndEnergyC}). This is a relevant scenario for truck-drone tandems and is required for synchronization. Since all flights are faster than \SI{180}{\second}, the drone must hover regardless of the speed. However, hovering time increases with increasing drone speed.

\begin{figure}
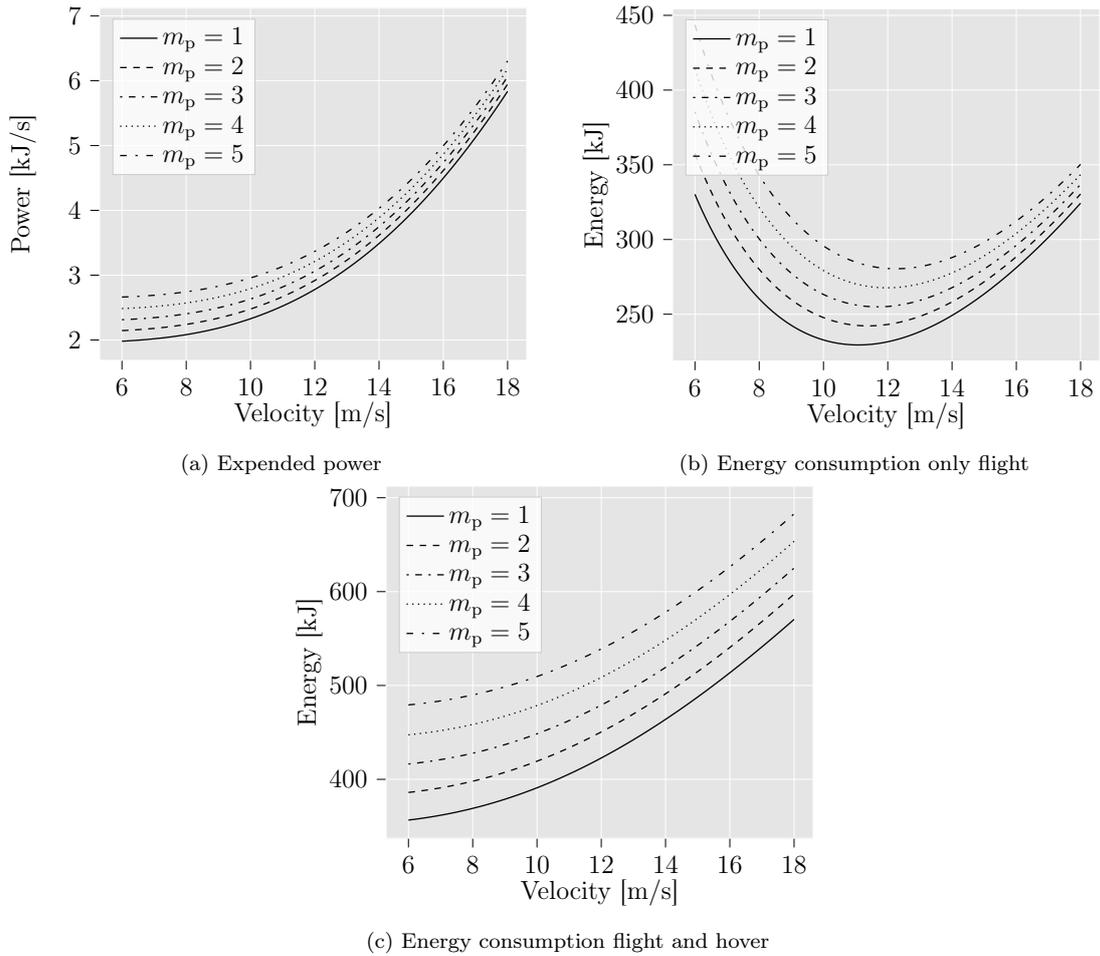

	\centering
	\begin{subfigure}{0.45\textwidth}
		\scalebox{.82}{\input{Figures/powerOverSpeed}}
		\subcaption{Expended power}
		\label{fig:powerAndEnergyA}
	\end{subfigure}
	\begin{subfigure}{0.45\textwidth}
		\scalebox{.82}{\input{Figures/energyOverSpeed}}
		\subcaption{Energy consumption only flight}
		\label{fig:powerAndEnergB}
	\end{subfigure}
	\begin{subfigure}{0.45\textwidth}
		\scalebox{.82}{\input{Figures/energyPlusHoverOverSpeed}}
		\subcaption{Energy consumption flight and hover}
		\label{fig:powerAndEnergyC}
	\end{subfigure}
	\caption{Expended power, energy consumption to fly 1000\,m, and energy consumption to fly 1000\,m plus hovering up to 180\,s are reached for different speeds and package masses $m_{\text{p}}$}
	\label{fig:powerAndEnergy}
\end{figure}

The expended power increases monotonously with the drone speed and is higher for larger package masses. In contrast, the energy consumption for steady flight initially decreases and then increases again with increasing speed; this is due to the trade-off between expended power and flight duration. At low speeds, less power is expended, but the flight duration is longer. In the reverse case, exactly the opposite is true; higher speeds expend more power, but the flight is faster. As a result, the range of the drone depends significantly on the selected speed, and faster drone speeds are not automatically better. Additionally, if waiting for the truck and hovering were not necessary in general, we could determine an optimal speed for a given package mass $m_{\text{p}}$ and exclude all slower speeds. Speeds slower than the optimum speed cause higher energy consumption plus longer flight duration. However, hovering might be necessary for truck-drone tandems. If we include energy consumption for hovering, slower speeds usually have a smaller total energy consumption, as shown in Figure \ref{fig:powerAndEnergyC}. 

Hence, faster drones may lead to faster deliveries but have a smaller range and higher energy consumption. The latter is especially important if the battery has to be recharged on the truck before the next flight and is not swapped. Therefore, it is essential to include the speed of the drones in the route-planning process of truck-drone tandems when their energy consumption is considered. 

\section{Mixed-integer linear programming formulation for the VRPD-DSS} \label{sec:problemDefinition}
\allowdisplaybreaks

\subsection{Notation}
\subsubsection{Sets and parameters}
We distinguish five different but partly overlapping sets of nodes. First, we denote the set of all customers by $C = \{1, \dots, c\}$. Not all customers can be served by a drone due, for example, to weight restrictions or customer preferences. Therefore, we introduce subset $\bar{C} \subseteq C$ as the set of all customers that can be served by a drone. Furthermore, we introduce nodes $0$ and $c+1$ as start depot and end depot for the same physical location. We then define $N = \{0\} \cup C \cup \{c+1\}$ as the set of all nodes, $N_0 = N \backslash \{c+1\}$ as the set of all departure nodes, and $N_+ = N \backslash \{0\}$ as the set of all arrival nodes.

A homogeneous fleet of truck-drone tandems $F$ is available to supply all customers. Each tandem $f \in F$ consists of a single truck $f$ and a set of drones $D$. The distance from node $i \in N$ to node $j \in N$ for a truck is denoted by $\delta^{\text{T}}_{ij}$ and the corresponding travel time by $\tau^{\text{T}}_{ij}$. The distance between two nodes $i$ and $j$ for the drone is represented by $\delta^{\text{D}}_{ij}$. In contrast to a truck, a drone can travel at different speeds $v \in V$, where $V$ is the set of possible speeds. We introduce $\tau^{\text{D},v}_{ij} = \delta^{\text{D}}_{ij} / v$ as the travel time of a drone from node $i$ to node $j$ at speed $v$. In addition to travel times, we consider service times $\tau_j^{\text{S,T}}$ and $\tau_j^{\text{S,D}}$ for truck and drone deliveries for each node $j \in N$. The amount of time needed to prepare a launch is represented by $\tau^{\text{L}}$. We also introduce the maximum amount of time a drone is allowed to hover before retrieval as $\tau^{\text{MH}}$ and the maximum time a truck is allowed to remain stationary at a node as $\tau^{\text{MS}}$. Both times can be limited to reflect more-realistic scenarios. The maximum duration of a route is denoted by $M$.

A drone flight is defined as triple $\left(i,j,k\right)$ with node $i \in N_0$ as the launch node, $j \in \bar{C}$ as the customer node, and $k \in N_+$ as the retrieval node. An operation $\left(i,j,k\right)^v$ represents the execution of the flight $\left(i,j,k\right)$ with speed $v$. Assuming that both legs of the flight ($i$ to $j$ and $j$ to $k$) are executed at the same speed $v$, we can determine the time $\tau^v_{ijk}$ of an operation with
\begin{equation}
\tau^v_{ijk} = \tau^{\text{D}, v}_{ij} + \tau_j^{\text{S,D}} + \tau^{\text{D}, v}_{jk}.
\end{equation}
The energy consumption of an operation can be computed in a similar manner. During the flight to the customer $j$, the drone must carry the package with mass $m_j$, whereas no package is transported on the way to retrieval node $k$. We assume that the drone hovers at customer $j$'s property to deliver the package and that the mass is constant ($m_j$) over the delivery time to take additional energy consumption, e.g., for using the winch, into account. Therefore, the energy consumption for operation $\left(i,j,k\right)^v$ corresponds to
\begin{equation}
e^v_{ijk} = \tau^{\text{D}, v}_{ij} \cdot P^{\text{F}}\left(m_j, v\right) + \tau_j^{\text{S,D}}  \cdot P^{\text{H}}\left( m_j \right) + \tau^{\text{D}, v}_{jk} \cdot P^{\text{F}}\left(0, v\right).
\end{equation}

The battery of a drone has a nominal energy of $E$. However, in order to increase its service life, a battery should usually not be fully discharged. Hence, we use $\epsilon$ as the maximum depth of discharge (DoD) in percent. A DoD of 0\% means the battery is fully charged, while at a DoD of 100\%, the battery is empty. Therefore, the maximum available energy is $\epsilon E$. Furthermore, it can be recharged with a fixed charging rate of $P^{\text{C}}$.

$W_v$ is the set of feasible drone operations for speed $v \in V$. Each drone speed $v$ leads to a different set of feasible operations since the speed has a large impact on the range. The set of all feasible drone operations is $W = \bigcup_{v \in V} W^v$. An operation $\left(i,j,k\right)^v$ is feasible only under three conditions: (i) all nodes have to be pairwise different; (ii) customer $j$ can be supplied by a drone, i.e., $j \in \bar{C}$; and (iii) the minimum energy consumption of operation $\left(i,j,k\right)^v$ does not exceed the maximum available energy $\epsilon E$ of the battery. In addition to the energy consumption $e^v_{ijk}$, we can take into account the minimum hovering time at retrieval node $k$ as the truck could arrive after the drone, although it travels directly from $i$ to $k$. Thus, operation $\left(i,j,k\right)^v$ is feasible for $i \in N_0, j \in \bar{C}, k \in N_+, i \neq j, i \neq k, j \neq k$, if
\begin{equation}
	e^{v}_{ijk} + \max(\tau^{\text{T}}_{ik} - \tau^v_{ijk}, 0) \cdot P^{\text{H}}(0) \leq \epsilon E.
\end{equation}

Finally, we define the cost parameters: (i) $\lambda$ as fuel cost per distance unit traveled by a truck, (ii) $\beta$ as cost per time unit of working time of a truck driver, and (iii) $\gamma$ as cost per energy unit expended by the drone.

\subsubsection{Decision variables}
Several decision variables are required to describe the problem as an MILP:
\begin{compactitem}
    \item $x_{ij}^f = 1$ if truck $f \in F$ drives directly from node $i \in N_0$ to node $j \in N_+$ and, otherwise, 0.
    \item $y_{ijk}^{fdv} = 1$ if drone $d$ of tandem $f \in F$ performs operation $\left( i,j,k \right)^v$ and, otherwise, 0.
    \item $q_{i}^f = 1$ if node $i \in N$ is visited by truck $f \in F$ and, otherwise, 0.
    \item $b_{ij}^f = 1$ if nodes $i,j \in C$ with $j>i$ are visited by truck $f \in F$ and, otherwise, 0.
    \item $u_i^f \in \mathbb{N}_0$ specifies the position of customer $i \in C$ on the route of truck $f \in F$. 
    \item $p_{ij}^f = 1$ if node $i \in N_0$ precedes nodes $j \in N_+$ in the tour of truck $f \in F$ and, otherwise, 0.
    \item $z_{lik}^{fd} = 1$ if drone $d \in D$ of tandem $f \in F$ performs a flight with launch node $i \in N_0$ and retrieval node $k \in N_+$ and truck $f$ visits node $l \in C$ in between and, otherwise, 0.
    \item $att_i^f \in \mathbb{R}_+$ represents the arrival time of truck $f \in F$ at node $i \in N$.
    \item $dtt_i^f \in \mathbb{R}_+$ represents the departure time of truck $f \in F$ from node $i \in N$.
    \item $atd_i^{fd} \in \mathbb{R}_+$ represents the arrival time of drone $d \in D$ of tandem $f \in F$ at node $i \in N$.
    \item $dtd_i^{fd} \in \mathbb{R}_+$ represents the departure time of drone $d \in D$ of tandem $f \in F$ from node $i \in N$.
    \item $htd_i^{fd} \in \mathbb{R}_+$  represents the amount of time that drone $d \in D$ of tandem $f \in F$ hovers at node $i \in N$.
    \item $ltd_i^{fd} \in \mathbb{R}_+$ represents the amount of time that is used for drone $d \in D$ of tandem $f \in F$ to be loaded at node $i \in N$.
    \item $r_i^{fd} \in \left[(1-\epsilon) E, E \right]$ represents the residual energy of drone $d \in D$ of tandem $f \in F$ when arriving at node $i \in N$ or at reunion with truck $f$ if $i$ is a retrieval node. 
    \item $w_{ij}^{fd} \in \left[ 0,1 \right]$ represents the share of travel time $\tau^{\text{T}}_{ij}$ on arc $\left(i,j\right)$ that is used by drone $d \in D$ of tandem $f \in F$ for recharging.
    \item $tec^{fd} \in \mathbb{R}_+$ represents the total energy consumption of drone $d$ of tandem $f$.
\end{compactitem}

The relationship between time variables of trucks and drones as well as energy-related variables of drones is shown in Figure~\ref{fig:time_and_energy}. Note that we omit the indices for truck and drone as we consider only one vehicle of each type in this example. 
\begin{figure}
	\centering
	\begin{tikzpicture}
\draw[thick, -Triangle] (-9pt,0) -- (10cm,0) node[font=\scriptsize,below left=3pt and -8pt]{Time};
\draw[thick, -Triangle] (-6pt,-3pt) -- (-6pt,5cm) node[font=\scriptsize,above left=3pt and -40pt]{Residual energy};
\draw (0 cm,3pt) -- (0 cm,-3pt);
\draw (4.75 cm,3pt) -- (4.75 cm,-3pt);
\draw (5.75 cm,3pt) -- (5.75 cm,-3pt);
\draw (6.5 cm,3pt) -- (6.5 cm,-3pt);
\draw (9 cm,3pt) -- (9 cm,-3pt);
\draw (-6pt, 4.5cm) -- (0pt, 4.5cm);
\draw[dashed] (-6pt, 1.50cm) -- (5.75cm, 1.50cm);
\draw[dashed](-6pt, 3cm) -- (9cm, 3cm);

\node[font=\scriptsize, text height=1.75ex, text depth=.5ex] at (0cm,+.5) {$dtt_i = dtd_i$};

\node[font=\scriptsize, text height=1.75ex, text depth=.5ex] at (4.75cm,+.5) {$atd_k$};

\node[font=\scriptsize, text height=1.75ex, text depth=.5ex] at (5.75cm,+.5) {$att_k$};
\node[font=\scriptsize, text height=1.75ex, text depth=.5ex] at (6.50cm,+.5) {$dt_k$};
\node[font=\scriptsize, text height=1.75ex, text depth=.5ex] at (9.00cm,+.5) {$att_l$};
\node[font=\scriptsize, text height=1.75ex, text depth=.5ex] at (-.4,4.50) {$r_i$};
\node[font=\scriptsize, text height=1.75ex, text depth=.5ex] at (-.4,1.50) {$r_k$};
\node[font=\scriptsize, text height=1.75ex, text depth=.5ex] at (-.4,3) {$r_l$};

\draw [thick, decorate, decoration={brace,amplitude=5pt}] (2,-0.2)  -- +(-2,0) 
node [black,midway,below=4pt, font=\scriptsize] {$\tau^{\text{D},v}_{ij}$};
\draw [thick, decorate, decoration={brace,amplitude=5pt}] (2.75,-0.2)  -- +(-0.75,0) 
node [black,midway,below=4pt, font=\scriptsize] {$\tau^{\text{S,D}}$};
\draw [thick, decorate, decoration={brace,amplitude=5pt}] (4.75,-0.2)  -- +(-2,0) 
node [black,midway,below=4pt, font=\scriptsize] {$\tau^{\text{D},v}_{jk}$};
\draw [thick, decorate, decoration={brace,amplitude=5pt}] (5.75,-0.2)  -- +(-1,0) 
node [black,midway,below=4pt, font=\scriptsize] {$htd_k$};
\draw [thick, decorate, decoration={brace,amplitude=5pt}] (6.50,-0.2)  -- +(-0.75,0) 
node [black,midway,below=4pt, font=\scriptsize] {$ltd_k$};
\draw [thick, decorate, decoration={brace,amplitude=5pt}] (4.75,-0.9)  -- +(-4.75,0) 
node [black,midway,below=4pt, font=\scriptsize] {$\tau^{\text{v}}_{ijk}$};

\draw [thick, decorate, decoration={brace,amplitude=5pt}] (5.75,-1.5)  -- +(-5.75,0) 
node [black,midway,below=4pt, font=\scriptsize] {$\tau^{\text{T}}_{ik}$};
\draw [thick, decorate, decoration={brace,amplitude=5pt}] (6.5,-1.5)  -- +(-.75,0) 
node [black,midway,below=4pt, font=\scriptsize] {$\tau^{\text{S,T}}$};
\draw [thick, decorate, decoration={brace,amplitude=5pt}] (9,-1.5)  -- +(-2.5,0) 
node [black,midway,below=4pt, font=\scriptsize] {$\tau^{\text{T}}_{kl}$};

\node[font=\scriptsize, text height=1.75ex, text depth=.5ex] at (-1cm,-.6) {Drone};
\node[font=\scriptsize, text height=1.75ex, text depth=.5ex] at (-1cm,-1.5) {Truck};

\draw [thick, decorate, decoration={brace,amplitude=5pt}] (-.6,1.75)  -- +(0,2.75)
node [black,midway,left=4pt, font=\scriptsize] {$e^{v}_{ijk}$};

\draw (0 cm, 4.5 cm) -- node[font=\scriptsize, text height=1.75ex, text depth=.5ex, above = 6pt ] {$P^{\text{F}}(m_j, v)$} (2 cm,3cm);
\draw (2 cm, 3 cm) -- node[font=\scriptsize, text height=1.75ex, text depth=.5ex, above = 6pt ] {$P^{\text{H}}(m_j)$} (2.75 cm, 2.75cm);
\draw (2.75 cm, 2.75 cm) -- node[font=\scriptsize, text height=1.75ex, text depth=.5ex, above = 3pt ] {$P^{\text{F}}(0, v)$} (4.75 cm, 1.75cm);
\draw (4.75 cm, 1.75 cm) -- node[font=\scriptsize, text height=1.75ex, text depth=.5ex, above = -0.5pt ] {$P^{\text{H}}(0)$} (5.75 cm, 1.50 cm);
\draw (5.75 cm, 1.50 cm) -- node[font=\scriptsize, text height=1.75ex, text depth=.5ex, above = -4.5pt] {$P^{\text{C}}$} (9 cm, 3.00 cm); 
\end{tikzpicture}
	\caption{Relation between variables for resources time and energy}
	\label{fig:time_and_energy}
\end{figure}
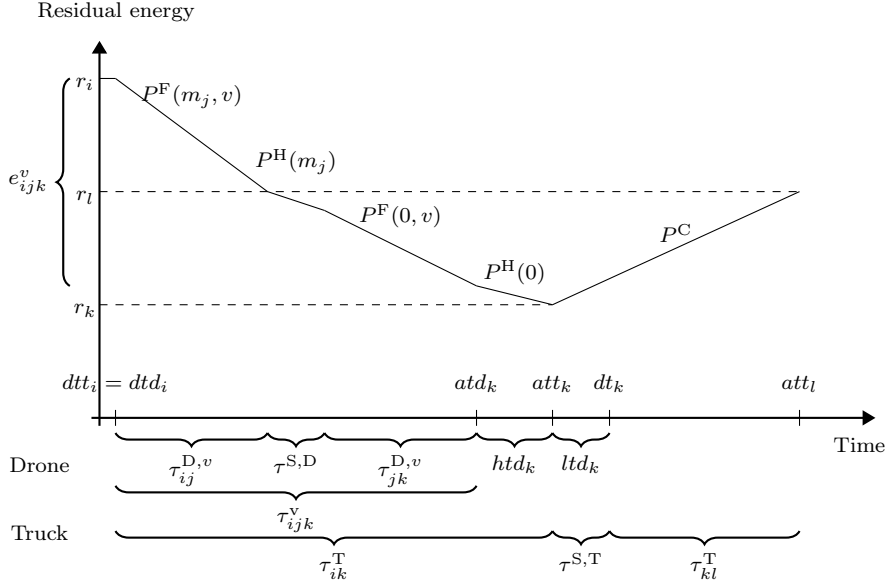
Both vehicles leave node $i$ at the same time ($dtt_i = dtd_i)$. The drone performs operation $\left( i,j,k \right)^v$ and arrives at retrieval node $k$ at time $atd_k = dtd_i + \tau^{v}_{ijk}$. Its residual energy decreases from $r_i$ to $r_i - e^{v}_{ijk}$ on arrival at node $k$. The truck arrives at node $k$ at $att_k = dt_i + \tau^{\text{T}}_{ik} > atd_k$. Thus, the drone must hover for $htd_k = att_k - atd_k$ time units. The residual energy of the drone at its reunion with the truck at node $k$ is $r_k = r_i - e^{v}_{ijk} - htd_k  P^{\text{H}}(0)$. After the arrival of the truck, customer $k$ is served by the driver. The drone is recharged on the truck for $ltd_k = \tau_k^{\text{S,T}}$ time units during the service to customer $k$. Finally, the truck departs from customer location $k$ at time $dt_k = att_k + \tau_k^{\text{S}}$ and travels to node $l$. The drone is recharged during the complete travel time $\tau^{\text{T}}_{kl}$ $\left(w_{kl} = 1 \right)$ while atop the truck, and its residual energy is $r_l = \tau^{\text{T}}_{kl} P^{\text{C}}$ when the tandem reaches node $l$.

\subsection{Model}
The VRPD-DSS can be formulated as an MILP with the notation and decision variables introduced above. To facilitate understanding, the objective function and the various groups of constraints are presented in several sections. 
\subsubsection{Objective function}
The objective function
\begin{align}
\min \lambda \sum_{f \in F} \sum_{i \in C}\sum_{j \in N_+} \delta^{\text{T}}_{ij}x^f_{ij} + \beta \sum_{f \in F} att_{c+1}^f + \gamma \sum_{f \in F} \sum_{d \in D} tec^{fd} \label{eq:objfunction}
\end{align}
minimizes the total operational costs. The first term of \eqref{eq:objfunction} corresponds to the fuel-consumption costs of the total distance traveled by all trucks. The second term represents the total working-time costs of all drivers, and the total energy costs of all drones are determined by the last term.

\subsubsection{Complete demand satisfaction}
Constraints
\begin{align}
& \sum_{f \in F} \left( q_{j}^f + \sum_{i \in N_0} \sum_{k \in N_+} \sum_{d \in D} \sum_{v \in V} y_{ijk}^{fdv} \right) = 1 \quad \forall j \in C 
\label{eq:completedemand}
\end{align}
guarantee that all packages must be delivered and ensure that each customer is visited only once by truck or drone.

\subsubsection{Truck routing}
We introduce constraints
\begin{align}
&\sum_{i \in N_0} x_{ij}^f = q_j^f \quad \forall ~ j \in N_+,  f \in F \label{eq:inflow}\\
&\sum_{j \in N_+} x_{ij}^f = q_i^f \quad \forall ~ i \in N_0, f \in F \label{eq:outflow}\\
&u_{i}^f - u_{j}^{f} + c \cdot x_{ij}^f + (c-2) \cdot x_{ji}^f \leq c - 1 \quad \forall ~ i,j \in C, f \in F \label{eq:liftedMTZ}\\
&c \cdot p_{ij}^f - \left(c-1\right) \leq u_j^f - u_i^f \quad \forall ~ i,j \in C, f \in F \label{eq:orderingtrucknodes1} \\
&p_{0,i}^f = q^f_i \quad \forall ~ i \in C, f\in F \label{eq:depotBeforeCustomer}\\
&p_{i,c+1}^f = q^f_i \quad \forall ~ i \in C, f \in F \label{eq:depotAfterCustomer}
\end{align}
to ensure feasible truck routes. 
Constraints \eqref{eq:inflow} and \eqref{eq:outflow} preserve the flow of a vehicle $f$.
Inequalities \eqref{eq:liftedMTZ} are lifted Miller--Tucker--Zemlin subtour elimination constraints. However, their primary purpose is not to prevent 
subtours but to correctly determine variables $u$.
Variables $u$ are used in constraints \eqref{eq:orderingtrucknodes1} to set precedence variables $p$. If $p^f_{ij} = 1$, then node 
$j$ is visited after node $i$ by truck $f$ and $u^f_j$ is larger than $u^f_i + 1$.
Equations \eqref{eq:depotBeforeCustomer} and \eqref{eq:depotAfterCustomer} guarantee that the depot precedes (node 0) and succeeds (node $c+1$) customer $i$ on the route of truck $f$ if and only if customer $i$ is visited by truck $f$.

Additionally, we ensure that either $p_{ij}^f$ or $p_{ji}^f$ equals 1 if and only if both nodes $i$ and $j$ are visited by the truck. 
Hence, we impose $p_{ij}^f + p_{ji}^f = q_{i}^f \cdot q_{j}^f ~ \forall i,j \in C, j > i$. As the right-hand side of this inequality is nonlinear, we use variables $b_{ij}^f$ and the following inequalities to linearize this relationship:
 \begin{align}
 &p_{ij}^f + p_{ji}^f = b_{ij}^f \quad \forall ~ i,j \in C, j > i, f \in F \label{eq:orderingtrucknodes2} \\
 &q_{i}^f \leq b_{ij}^f \quad \forall ~ i,j \in C,j > i, f \in F \label{eq:linearization1} \\
 &q_{j}^f \leq b_{ij}^f \quad \forall ~ i,j \in C,j > i, f \in F \label{eq:linearization2}\\
 &q_{i}^f + q_{j}^f \leq 1 + b_{ij}^f \quad \forall ~ i,j \in C, j > i, f \in F \label{eq:linearization3}
\end{align}
 
\subsubsection{Coordination of drone actions}
A drone can perform various actions. Considering our assumptions, it can
\begin{compactitem}
	\item start or end a flight at a node,
	\item be airborne while the truck visits a node,
	\item recharge while the truck is traveling from one node to another,
	\item recharge at a node, or
	\item idle on the truck either at a node or on an arc.
\end{compactitem} 
However, it can never perform multiple activities at the same time, and its actions must be coordinated with the activities of its truck. Since times when the drone is idle do not need to be modeled separately and recharging at a node is included in the energy-consumption constraints in Section~\ref{sec:constraintsEnergyConsumption}, we need to introduce only constraints
\begin{align}
&\sum_{i \in N_0} \sum_{j \in C} \sum_{v \in V} y_{ijk}^{fdv} \leq q_k^f \quad \forall ~ k \in N_+, f \in F, d \in D \label{eq:endnodedronetruck} \\
&w_{ik}^{fd} \leq x_{ik}^f \quad \forall ~ i,k \in N, f \in F, d \in D \label{eq:loadingOnArc} \\
& z_{lik}^{fd} \geq p_{il}^f + p_{lk}^f + \sum_{j \in C} \sum_{v \in V} y_{ijk}^{fdv} - 2 \quad \forall ~ l \in C, i \in N_0, k \in N_+, f \in F, d 
\in D \label{eq:droneInAir}
\end{align}
to model the coordination of the remaining activities.
Constraints \eqref{eq:endnodedronetruck} assure that the retrieval node of a drone flight has to be visited by the truck.
Constraints \eqref{eq:loadingOnArc} guarantee that a charging activity of drone $d$ on arc $\left(i,k\right)$ may occur only if truck $f$ travels from $i$ to $k$.
Constraints \eqref{eq:droneInAir} determine whether drone $d$ is in the air at node $l$ while the truck is visiting $l$. Thus, these constraints determine whether a drone is available to start an activity at node $l$ or not. Consider Figure \ref{fig:exampleDroneFlight} for a better understanding. Truck $f$ visits customer $l$ between nodes $i$ and $k$; therefore, $p_{il}^f = 1$ and $p_{lk}^f = 1$. At the same time, drone $d$ starts a flight with any speed $v$ at node $i$, visits a customer $j$, and is retrieved by the truck at $k$. Thus, drone $d$ is in the air at node $l$ and constraints \eqref{eq:droneInAir} enforce  $z_{lik}^{fd} = 1$. 
\begin{figure}[ht]
	\centering
		\begin{tikzpicture}
		\node[circle, draw] (i) at (0,0) {$i$};
		\node[circle, draw] (l) at (4,0) {$l$};
		\node[circle, draw] (k) at (8,0) {$k$};
		\node[circle, draw] (j) at (4,2) {$j$};
		\node (y1) at (4,2.7) {$\sum\limits_{j \in C} \sum\limits_{v \in V} y_{ijk}^{fdv} = 1$};
		\node (z) at (4, -0.7) {$z_{lik}^{fd} \geq 1 + 1 + 1 - 2 = 1 $};
		\graph{
			(i) ->[dotted, "$p_{il}^f = 1$"] (l) ->[dotted, "$p_{lk}^f = 1$"] (k);
			(i) ->[bend left = 20, dashed] (j) ->[bend left = 20, dashed] (k);
	};
	\end{tikzpicture}
	\caption{Visualization of constraints \eqref{eq:droneInAir} that determine if drone $d$ of tandem $f$ is in air at node $l$.}
	\label{fig:exampleDroneFlight}
\end{figure}
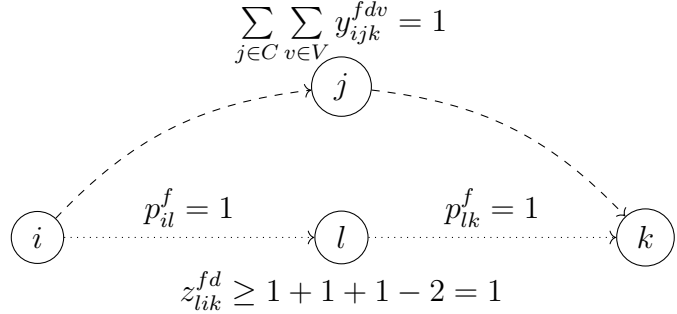

Now, multiple actions that take place simultaneously can be prevented with
\begin{align}
&\sum_{n \in N_+} w_{lm}^{fd} + \sum_{v \in V} \sum_{m \in C}\sum_{n \in N_+} y_{lmn}^{fdv} \leq q_l^f - \sum_{i \in N_0}\sum_{k \in N_+} z_{lik}^{fd}
\quad \forall ~ l \in C, f \in F, d \in D. \label{eq:onlyOneActivity}
\end{align}
Constraints \eqref{eq:onlyOneActivity} guarantee that drone $d$ can either be charged on an arc leaving node $l$ or can start an operation from node $l$. However, these actions are possible only if node $l$ is visited by truck $f$ and drone $d$ is not already in the air at node $l$. 

\subsubsection{Temporal synchronization between trucks and drones}
We introduce constraints
\begin{align}
&att_k^f \geq dtt_i^f + \tau^{\text{T}}_{ik} \cdot  x_{ik}^f - M_i \left(1-x_{ik}^f\right) \quad \forall ~ i \in N_0, k \in N_+, f \in F \label{eq:timetruckLB} \\
&\begin{multlined}
    atd_k^{fd} \geq dtd_i^{fd} + \sum_{v \in V} \sum_{j \in C} \tau^v_{ijk}  y_{ijk}^{fdv} - M_i \left(1 - \sum_{v \in V} \sum_{j\in C} y_{ijk}^{fdv} \right) \\ \forall ~ f \in F, d \in D, i \in N_0, k \in N_+ \label{eq:arrivalDroneCustomerLB} 
\end{multlined}\\
&\begin{multlined}
atd_k^{fd} \leq dtd_i^{fd} + \sum_{v \in V} \sum_{j \in C} \tau^v_{ijk} y_{ijk}^{fdv} + M_i \left(1 - \sum_{v \in V} \sum_{j \in C} y_{ijk}^{fdv} \right) \\ \forall ~ f \in F, d \in D, i \in N_0, k \in N_+  \label{eq:arrivalDroneCustomerUB}
\end{multlined}\\
&htd_i^{fd} \geq att_i^f - atd_{i}^{fd} \quad \forall ~ i \in C, f \in F, d \in D  \label{eq:hoverTimeDrone} \\
&htd_i^{fd} \leq q_i^f \tau_i^{\text{MH}} \quad \forall ~ i \in C, f \in F, d \in D \label{eq:lb_hovering}\\
&ltd_i^{fd} \leq q_i^f \tau_i^{\text{MS}} \quad \forall ~ i \in C, f \in F, d \in D \label{eq:lb_loading}
\end{align}
to synchronize the activities of trucks and their associated drones with respect to time.
Constraints \eqref{eq:timetruckLB} bound the arrival time $att_k^f$ of truck $f$ at node $k$ if truck $f$ travels directly from node $i$ to node $j$, where $M_i = M - \tau^{\text{T}}_{i,c+1}$ is the latest possible departure time from node $i$. Constraints \eqref{eq:arrivalDroneCustomerLB} and \eqref{eq:arrivalDroneCustomerUB} set the arrival time $atd_k^{fd}$ of drone $d$ belonging to tandem $f$ at reunification node $k$ if it performs operation $\left(i,j,k\right)^v$. Hover time $htd_k^{fd}$ of drone $d$ at node $k$ is defined by constraints \eqref{eq:hoverTimeDrone}. In case drone $d$ arrives before its corresponding truck $f$ at node $i$ $\left(atd_i^{fd} < att_{i}^{f}\right)$, it is equal to the difference between the truck arrival time $att_i^f$ and the drone arrival time $atd_{i}^{fd}$; otherwise, it is 0. Inequalities \eqref{eq:lb_hovering} and \eqref{eq:lb_loading} can be used to limit the maximum time a drone is allowed to hover at node $i$ and the maximum time a drone can be recharged at customer node $i$. They also ensure that drone $d$ of tandem $f$ can only hover or be recharged at customer $i$ if truck $f$ visits customer $i$.

Constraints
\begin{align}
	&dtt_i^{f} \geq att_i^f + \tau_i^{\text{S,T}} \quad \forall ~ i \in N, f \in F \label{eq:departureTimeTandemTruck}\\
	&dtt_i^f \geq dtd_i^{fd} \quad \forall ~ i \in N, f \in F, d \in D\label{eq:departureTimeTandemDrone} \\		
	&dtt_i^f \leq att_i^f + \tau^{\text{MS}} \quad \forall ~ i \in C, f \in F \label{eq:maxStationaryTruck} \\
    &dtd_i^{fd} \geq atd_i^{fd} + htd_i^{fd} + ltd_i^{fd} + \tau^{\text{L}} \sum_{j \in C'} \sum_{k \in N_+} \sum_{v \in V} y_{ijk}^{fdv} \quad \forall ~ i \in N, f \in F, d \in D \label{eq:departureTimeDrone}
\end{align}
set the departure times of trucks and drones.
The earliest departure time of truck $f$ from node $i$ is determined by constraints \eqref{eq:departureTimeTandemTruck} and \eqref{eq:departureTimeTandemDrone}. In addition, constraints \eqref{eq:maxStationaryTruck} limit the maximum time a truck can remain stationary at a node.
The departure time of a drone is determined with constraints \eqref{eq:departureTimeDrone}. Drone $d$ must not depart from node $i$ before it has finished hovering and loading and is prepared for launch if it starts an operation at node $i$. Thus, all truck- and drone-related activities at a node must be completed before the truck can leave that node.

\subsubsection{Energy consumption of drones} \label{sec:constraintsEnergyConsumption}
The total energy consumption $tec^{fd}$ of drone $d$ belonging to tandem $f$ is determined by constraints
\begin{align}
&tec^{fd} = \sum_{i \in N_0} \sum_{j \in C} \sum_{k \in N_+} \sum_{v \in V} e_{ijk}^{v}  y_{ijk}^{fdv} + P^{\text{H}}(0) \sum_{i \in C} htd_{i}^{fd}   \quad \forall ~ f \in F, d \in D \label{eq:total_energy_consumption_use}
\end{align}
and is equal to the energy expended for all flights plus the energy expended for hovering. The residual energy of a drone is computed by constraints
\begin{align}
&\begin{multlined} 
r_k^{fd} \leq r_i^{fd} + ltd_i^{fd} P^{\text{C}} - \sum_{v \in V} \sum_{j \in C} e^v_{ijk} y_{ijk}^{fdv} - htd_k^{fd} P^{\text{H}}(0) + E \left(1 - \sum_{v \in V} \sum_{j \in C} y_{ijk}^{fdv}\right) \\ \forall ~ f \in F, d \in D, i \in N_0, k \in N_+ \label{eq:energyFlight}
\end{multlined} \\
&r_i^{fd} + ltd_i^{fd} P^{\text{C}} \leq E \quad \forall ~ i \in N_0, f \in F, d \in D \label{eq:energyDeparture}  \\
&r_j^{fd} \leq r_{i}^{fd} + ltd_i^{fd} P^{\text{C}} + \tau^{\text{T}}_{ij} P^{\text{C}} w_{ij}^{fd} + E \left( 1 - x_{ij}^f \right) \quad \forall ~ i \in N_0, j \in N_+, f \in F, d \in D. \label{eq:energyLoading}
\end{align}
Constraints \eqref{eq:energyFlight} restrict the residual energy of drone $d$ when reuniting with truck $f$ at node $k$ if a flight is performed between nodes $i$ and $k$. The residual energy at reunification is the residual energy at departure from node $i$ minus the expended energy. The residual energy at departure from node $i$ consists of the residual energy at arrival $r_i^{fd}$ plus the recharged energy while stationary at node $i$. The expended energy comprises the energy $e^v_{ijk}$ for performing operation $\left(i,j,k\right)^v$ and the energy needed for hovering at reunification node $k$. 
Constraints \eqref{eq:energyDeparture} limit the residual energy at departure to the maximum value $E$ since this is not guaranteed by constraints \eqref{eq:energyFlight}. 
Finally, constraints \eqref{eq:energyLoading} represent the reloading of drone $d$ while traveling on truck $f$ from node $i$ to node $j$. However, this is valid only if truck $f$ uses arc $\left(i,j\right)$.

\section{Model strengthening}
\subsection{Preprocessing}
\label{sec:Preprocessing}
\subsubsection{Elimination of dominated drone operations}
Identifying drone speeds that will never be used in an optimal solution for a flight $\left(i,j,k\right)$ and eliminating them can substantially reduce the number of possible drone operations. Consequently, the number of variables in our model is also reduced, thereby improving the performance. In general, we have to consider the different resources time and energy consumption. An operation is referred to as dominated by another operation for the same flight if is not beneficial with respect to time or to energy consumption.

\label{sec:elimination_drone_flights}
\paragraph{General dominance rules}
\begin{prop} \label{prop:generalRule}
For flight $\left(i,j,k\right) \in P$, operation $\left(i,j,k\right)^v$ is dominated by operation $\left(i,j,k\right)^{s}$ with $s > v$ and $v,s \in V$, if 
\begin{equation}
    e^{s}_{ijk} + \left(\tau^v_{ijk} - \tau^s_{ijk}\right) P^{\text{H}}(0) \leq e^v_{ijk}. \label{eq:general_rule}
\end{equation}
\end{prop}
\begin{proof}
 Operation $\left(i,j,k\right)^{s}$ is faster than operation $\left(i,j,k\right)^v$ because $\tau^{s}_{ijk} < \tau^v_{ijk}$ if ${s} > v$. However, since the truck can arrive after the drone, we can state only that operation $\left(i,j,k\right)^{s}$ is always at least as good as operation $\left(i,j,k\right)^v$ with respect to time. 

The energy consumption of operation $\left(i,j,k\right)^v$ is $e^v_{ijk}$. Since the truck can arrive after the drone, we also have to consider the maximum additional hover time $\left(\tau^v_{ijk} - \tau^{s}_{ijk}\right)$ to reach the same point in time with operation $\left(i,j,k\right)^{s}$ as with operation $\left(i,j,k\right)^v$. After that, the energy expended by hovering is the same for both speeds. Therefore, the maximum energy consumption of operation $\left(i,j,k\right)^{s}$ to reach the same point in time as operation $\left(i,j,k\right)^v$ is 
\begin{equation}
e^{s}_{ijk} + \left(\tau^v_{ijk} - \tau^s_{ijk}\right) P^{\text{H}}(0). \label{eq:energyFaster}
\end{equation}
Thus, operation $\left(i,j,k\right)^{s}$ dominates operation $\left(i,j,k\right)^v$, if \eqref{eq:general_rule} is true.
\end{proof}

However, due to the drone power-usage models for flight and hover mode introduced in Section \ref{sec:energy_use_model}, it is unlikely that the dominance rule in Proposition~\ref{prop:generalRule} applies. This also highlights the trade-off between energy consumption and execution time. Nevertheless, we are able to construct two special cases to eliminate dominated operations. In the following, we first show how a slower operation can dominate a faster one. Secondly, we demonstrate the reverse case, where faster operations are superior to slower operations.

\paragraph{Elimination of operations with faster speeds}
\begin{prop}
For flight $\left(i,j,k\right) \in P$, operation $\left(i,j,k\right)^s$ is dominated by operation $\left(i,j,k\right)^v$ with $s > v$ and $v,s \in V$ if
\begin{equation}
	\tau^v_{ijk} \leq \tau^{\text{T}}_{ik} \land e^v_{ijk} \leq e^{s}_{ijk} + \left( \tau^v_{ijk} - \tau^{s}_{ijk} \right) \cdot 
	P^{\text{H}}\left( 0 \right). \label{eq:faster_rule}
\end{equation}
\end{prop}
\begin{proof}
As stated above, operation $\left(i,j,k\right)^s$ is always at least as good with respect to time as $\left(i,j,k\right)^v$. However, if $\tau^v_{ijk} \leq \tau^{\text{T}}_{ik}$, then there is no benefit in using the faster speed $s$. Thus, both operations are equal with respect to time. Now, we can eliminate operation $\left(i,j,k\right)^s$ if the energy consumption is at least as high as the energy consumption of the slower operation $\left(i,j,k\right)^v$. Analogous to the general case, the energy consumption of operation $\left(i,j,k\right)^v$ is $e^v_{ijk}$, and the energy consumption of operation $\left(i,j,k\right)^s$ to reach the same point in time as operation $\left(i,j,k\right)^v$ can be determined with \eqref{eq:energyFaster}. Thus, operation $\left(i,j,k\right)^s$ is dominated by operation $\left(i,j,k\right)^v$, if \eqref{eq:faster_rule} is true.
\end{proof}

\paragraph{Elimination of operations with slower speeds}
\begin{prop}
For flight $\left(i,j,k\right) \in P$, operation $\left(i,j,k\right)^v$ is dominated by operation $\left(i,j,k\right)^s$ with $s > v$ and $v,s \in V$ if
\begin{align}
 &\neg \exists l \in C, l \neq j \text{ s.t. } e^v_{ijk} + \max \left( \tau_{il}^{\text{T}} + \tau_l^{\text{S,T}} + \tau^{\text{T}}_{lk} - \tau^v_{ijk}, 0\right) \cdot P^{\text{H}}(0) \leq \epsilon E \label{eq:noDetourV} \\
\land &\neg \exists l \in C, l \neq j \text{ s.t. } e^s_{ijk} + \max \left( \tau^{\text{T}}_{il} + \tau_l^{\text{S,T}} + \tau^{\text{T}}_{lk} - \tau^s_{ijk}, 0\right) \cdot P^{\text{H}}(0) \leq \epsilon E \label{eq:noDetourS} \\
   \land &\tau^v_{ijk} > \tau^{\text{T}}_{ik} \land e^{s}_{ijk} + \max \left( \tau^{\text{T}}_{ik} - \tau^s_{ijk}, 0 \right) \cdot 
P^{\text{H}}\left( 0 \right) \leq e^v_{ijk} \label{eq:slower_rule}.
\end{align}
\end{prop}
\begin{proof}
Conditions \eqref{eq:noDetourV} and \eqref{eq:noDetourS} ensure that operations $\left(i,j,k\right)^v$ and $\left(i,j,k\right)^s$ require a direct trip of the truck from the launch node $i$ to the retrieval node $k$. Here, a direct trip is necessary if the truck cannot serve a customer $l \in C$ between $i$ and $k$ since this detour via $l$ would increase the hover time of the drone at $k$, resulting in energy consumption that is too high. Taking this special case into account, the amount of hover time is known as the truck travels directly from the launch to the retrieval node. Therefore, we are able to determine the expended energy, including hovering, exactly. In addition, we assume that the drone always arrives after the truck if operation $\left(i,j,k\right)^v$ is performed $\left( \tau^v_{ijk} > \tau^{\text{T}}_{ik} \right)$; hence, it never needs to hover. This also means that, in contrast to the general rule, if the drone performs operation $\left(i,j,k\right)^s$, it can always be retrieved by the truck before reaching the same point in time as operation $\left(i,j,k\right)^v$. Thus, the maximum additional hover time is $\max \left( \tau^{\text{T}}_{ik} - \tau^s_{ijk}, 0 \right)$ and operation $\left(i,j,k\right)^s$ dominates operation $\left(i,j,k\right)^v$ if conditions \eqref{eq:noDetourV} -- \eqref{eq:slower_rule} hold true.
\end{proof}

\subsubsection{Elimination of variables z}
 Different modeling approaches for the VRPD prohibit, in different ways, the launch of a drone when it is already in flight. We introduce variables $z_{lik}^{fd}$ to check whether drone $d$ of tandem $f$ performs a flight from $i$ to $k$ and is, therefore, not available at node $l$. This leads to a large number of variables for larger instances. However, we can eliminate several unnecessary variables to reduce the problem size without excluding any optimal solutions.
\begin{prop}
Variables $z_{lik}^{fd}$ for three pairwise different nodes $l \in C, i \in N_0, k \in N_+$ can be eliminated for all $f \in F, d \in D$ if there is no drone operation with launch node $i$ and retrieval node $k$ or 
\begin{align}
 \neg \exists \left(i,j,k\right)^v \in W_v, v \in V, l \neq j \text{ s.t. } e^v_{ijk} + \max \left( \tau^{\text{T}}_{il} + \tau_{l}^{\textrm{S,T}} + \tau^{\text{T}}_{lk} - \tau^v_{ijk}, 0\right) \cdot P^{\text{H}}(0) \leq \epsilon E. \label{eq:noDetourL}
\end{align}
\end{prop}
\begin{proof}
Following the definition of the variables, it is obvious that $z_{lik}^{fd}$ can be eliminated if there is no drone operation with launch node $i$ and retrieval node $k$. Condition \eqref{eq:noDetourL} states that there is no feasible operation with launch node $i$ and retrieval node $k$ if the truck performs a detour via node $l$, since the energy consumption of the operation plus the additional energy consumption while hovering at node $k$ exceeds the drone's available energy. Thus, the truck cannot visit node $l$ between $i$ and $k$ if a drone performs any operation with launch node $i$ and retrieval node $k$ and variables $z_{lik}^{fd}$ can be eliminated.
\end{proof}
\subsection{Valid inequalities} \label{sec:valid_inequalities}
Most of the valid inequalities used in this paper are similar to the valid inequalities introduced in \cite{tamkef-2021}. Since they are explained in detail there, we refer to that work for a more detailed discussion.

\subsubsection{Lower bounds on arrival and departure times}\label{sec:description_lowerBounds}
The lower bounds on arrival times at nodes and departure times from nodes are modified in comparison to \cite{tamkef-2021} to include the additional aspects considered in this paper. However, the operating principle is similar. The following inequalities set lower bounds on the completion time of a truck $f$ and a drone $d$ belonging to $f$:
\begin{align}
&att_{c+1}^f \geq \sum_{i \in N_0} \sum_{j \in N_+} \left( \tau^{\text{T}}_{ij} + \tau_j^{\text{S,T}} \right) x_{ij}^f \quad \forall f \in F \label{eq:lbArrivalTimeTruck} \\
&\begin{multlined} 
atd_{c+1}^{fd} \geq \sum_{i \in N_0} \sum_{j \in C} \sum_{k \in N_+} \sum_{v \in V} \left( \tau^{\text{L}} + \tau^{v}_{ijk} \right) y_{ijk}^{fdv} + \sum_{i \in N} \left[ htd_i^{fd} + ltd_i^{fd} \right] + \sum_{i \in N_0} \sum_{j \in N_+} w_{ij}^{fd} \tau^{\text{T}}_{ij} \\
\quad \forall f \in F, d \in D \label{eq:lbArrivalTimeDrone}
\end{multlined}\\
&att_{c+1}^f \geq dtt_i^f + \sum_{k \in N_+}\left[ \left( \tau^{\text{T}}_{ik}  + \tau_k^{\text{S,T}} + \tau^{\text{T}}_{k,c+1} \right) x_{ik}^f \right] \quad \forall ~ i \in C, f \in F \label{eq:lbCompletionTimeTandem}.
\end{align}
 The first two consider the total active time of the vehicles. The active time of a truck consists of travel and service times \eqref{eq:lbArrivalTimeTruck}. In addition to these, hovering and recharging times at nodes and on arcs must also be taken into consideration for drones \eqref{eq:lbArrivalTimeDrone}. In contrast to the first two inequalities, inequalities \eqref{eq:lbCompletionTimeTandem} determine the completion time based on the minimum travel time from a customer $i$ via another node $k$ back to the depot. If truck $f$ travels directly from $i$ to $k$, then the earliest arrival time at the depot is the departure time at node $i$, plus the travel time from $i$ to $k$, the service time at node $k$, and the travel time from node $k$ to the depot.
\begin{align}
&att_{k}^f \geq \sum_{i \in N_0} \left( \tau^{\text{T}}_{0,i} + \tau_i^{\text{S,T}} + \tau^{\text{T}}_{ik} \right) x_{ik}^f \quad \forall k \in C, f \in F \label{eq:lbDepartureTimeTandemNodeTruck}\\ 
&dtd_{k}^{fd} \geq \sum_{i \in N_0} \sum_{j \in C} \sum_{v \in V} \left(\tau^{\text{T}}_{0,i} + \tau^{\text{L}} + \tau^{v}_{ijk} \right) y_{ijk}^{fdv} + htd_{k}^{fd} + ltd_{k}^{fd} \quad \forall k \in C, f \in F, d \in D. \label{eq:lbDepartureTimeTandemNodeDrone}
\end{align}
Inequalities \eqref{eq:lbDepartureTimeTandemNodeTruck} establish lower bounds on the arrival time at a customer $k$. As in \eqref{eq:lbCompletionTimeTandem}, the detour via another node is considered, but now the truck starts at the depot and travels directly to detour node $i$. Inequalities \eqref{eq:lbDepartureTimeTandemNodeDrone} set lower bounds on the departure time of drone $d$ associated with truck $f$ at node $k$. Drone $d$ travels atop truck $f$ from the depot to detour node $i$ and performs an operation with retrieval node $k$

\subsubsection{Problem-specific cuts} \label{sec:description_cuts}
In addition to the lower bounds on arrival and departure times, we use the VRPD-specific cuts introduced in \cite{tamkef-2021}:
\begin{align}
&\sum_{i \in C} \sum_{f \in F} q_i^f \geq  \frac{|C| - |D| \cdot |F|}{|D| + 1} \label{eq:lb_truck_visits} \\
&x_{ik}^f \leq p_{ik}^f \quad \forall i \in N_0, k \in N_+, f \in F \label{eq:precedenceCutsTruck}\\
&\sum_{j \in C} \sum_{v \in V} y_{ijk}^{fdv} \leq p_{ik}^f \quad \forall i \in N_0, k \in N_+, f \in F, d \in D \label{eq:precedenceCutsDrone} \\
&x_{0c+1}^f + q_j^f \leq 1 \quad \forall j \in C f \in F \label{eq:noArtificialTrip}.
\end{align}
First, we set a lower bound on the number of customers that can be visited by all trucks with inequality \eqref{eq:lb_truck_visits}. Inequalities \eqref{eq:precedenceCutsTruck} state that, if truck $f$ travels directly from node $i$ to node $k$, then $i$ has to precede $k$ in the route of $f$. Inequalities \eqref{eq:precedenceCutsDrone} ensure that, if drone $d$ performs any flight with launch node $i$ and retrieval node $k$, then $i$ must be visited before $k$ by truck $f$. Inequalities \eqref{eq:noArtificialTrip} prohibit the artificial trip between the two depot nodes $0$ and $c+1$ if any customer is visited by the truck.

Furthermore, we use the extended subtour elimination constraints (ESECs)
\begin{align}
\begin{multlined}
\sum_{f \in F} \left[ \sum_{i \in S} \sum_{j \in S} x_{ij}^f  + \sum_{i \in \bar{S}} \sum_{j \in S} \sum_{k \in S} \sum_{d \in D} \sum_{v \in V} y_{ijk}^{fdv} \right. 
\\ \left. + \sum_{i \in S} \sum_{j \in S} \sum_{k \in \bar{S}} \sum_{d \in D} \sum_{v \in V} y_{ijk}^{fdv} + \sum_{i \in S} \sum_{j \in S} \sum_{k \in S} \sum_{d \in D} \sum_{v \in V} y_{ijk}^{fdv} \right] \leq |S| - 1 \quad \forall ~ S \subseteq C, |C| \geq 2 \label{eq:ESEC}
\end{multlined}
\end{align}
introduced in \cite{tamkef-2021} as well. Since there is an exponential number of ESECs, we cannot add them at the beginning but, rather, have to detect violated cuts during the optimization. Therefore, we use the separation algorithm presented in \cite{tamkef-2021}.

Finally, we introduce the following new cuts, which have been proven to be useful:
\begin{align}
&\sum_{j \in C} \sum_{v \in V} y_{ijk}^{fdv} \leq \sum_{l \in C} z_{lik}^{fd} + x_{ik}^f \quad \forall ~ i \in N_0, k \in N_+, f \in F, d \in D  \label{eq:visit_during_flight} \\
&r_i^{fd} + ltd_i^{fd} P^{\text{C}} - \left(1 - \epsilon \right) E \geq \sum_{v \in V} \sum_{j \in \bar{C}} \sum_{k \in N_+}  e^v_{ijk} y_{ijk}^{fdv} \quad \forall ~ i \in C, f \in F, d \in D \label{eq:lb_available_energy} \\
&tec^{fd} =  \sum_{i \in N_0} \sum_{j \in N_+} w_{ij}^{fd} \tau^{\text{T}}_{ij} P^{\text{C}} + 
\sum_{i \in C} ltd_{i}^{fd} P^{\text{C}} + \left(r_0^{fd} - r_{c+1}^{fd} \right) \quad \forall ~ f \in F, d \in D \label{eq:total_energy_consumption_recharge}.
\end{align}
Inequalities \eqref{eq:visit_during_flight} state that, if drone $d$ of tandem $f$ performs any operation with launch node $i$ and retrieval node $k$, then it has to be in the air at any node $l$ visited by truck $f$ between $i$ and $k$, or truck $f$ has to travel directly from $i$ to $k$. Inequalities \eqref{eq:lb_available_energy} set a lower bound on the available energy, consisting of the residual energy on arrival and the recharged energy while stationary, at node $i$. The available energy must be sufficient for a drone operation starting at node $i$. Finally, equations \eqref{eq:total_energy_consumption_recharge} represent a second variant to determine the total energy consumption of a drone. Constraints \eqref{eq:total_energy_consumption_use} take into account the energy used for flying and hovering. In contrast, equations \eqref{eq:total_energy_consumption_recharge} consider the energy that is used to recharge the battery of a drone. However, the battery need not be fully charged at the end of the tour. Therefore, we have to additionally consider the difference between the residual energy at the beginning and at the end to determine the drone's total energy consumption.

\section{Computational studies} \label{sec:ComputationalStudies}
The algorithm is implemented in C{\#} with .NET Framework 4.6.1 and Gurobi 9.0 is used as the MILP solver. All tests are performed on a Windows Server 2012 R2 with Intel(R) Xeon(R) CPU E5-4627 v2 \@ 3.3 GHz processors with 32 cores and 768GB RAM. We use 12 cores to solve each instance, and the memory consumption is very low. As in \cite{tamkef-2021}, extended subtour elimination constraints \eqref{eq:ESEC} are not added at every node of the branch-and-bound tree. Here, they are added at every 100th node.
\subsection{Generation of real-world rural-area test instances}
We generate test instances that represent a real-world, rural-area based scenario for the use of truck-drone tandems to test our approach and gain managerial insights. All instances are created with Python 3.7 and are available at \cite{tamkef-2021b}. 
\paragraph{Depot and customer locations}
The basis of all test instances is an rectangular-shaped area approximately 20\,km by 30\,km located in Minnehaha County, South Dakota, USA. We have selected approximately 700 possible customer locations and a UPS Customer Center in Sioux Falls as the depot. The map in Figure~\ref{fig:map_customers} shows the distribution of the selected customer locations (dots) and the depot (triangle). To create a single instance, we randomly select $|C|$ customers out of all customer locations.
\begin{figure}[ht]
	\centering
	\includegraphics[scale = 0.75]{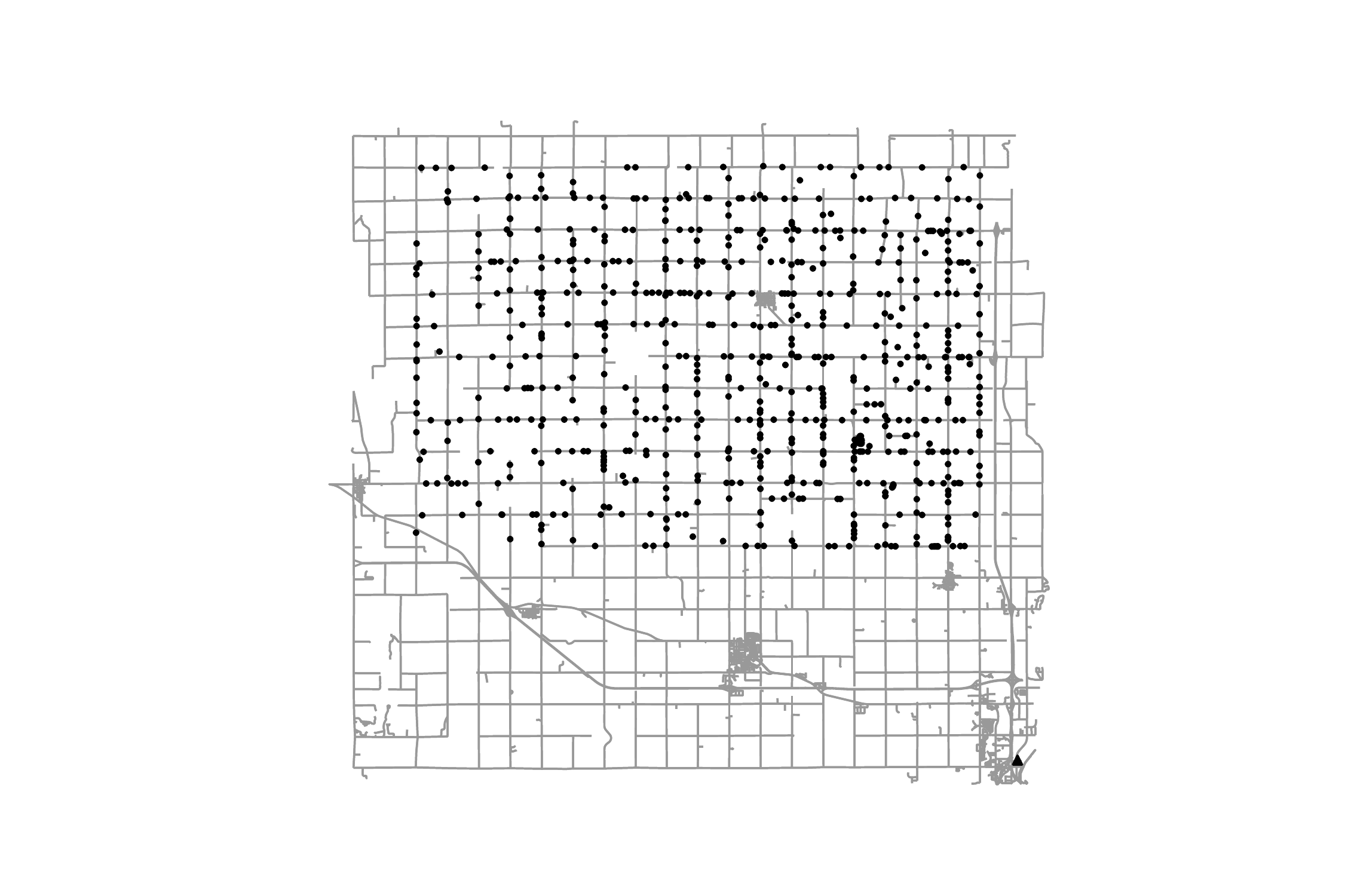}
	\caption{Distribution of all possible customers (dots) and the depot (triangle in the lower right corner)}
	\label{fig:map_customers}
\end{figure}
\paragraph{Specifications of drone model and battery}
As in the examples in Section~\ref{sec:energy_use_model}, we use the octocopter model presented in \cite{stolaroffj-2018}. We assume an overall power efficiency of the drone $\eta = 0.7$ and a safety coefficient $\sigma = 0.2$. Similar to \cite{sacramentod-2019}, we assume that all drones can transport packages weighing up to 5\,kg. 
In addition, we use an existing lithium polymer (LiPo) battery from Grepow Inc. to power the drone \cite{grepow-2019}. Since large drones have higher energy consumption, they also require large batteries. The LiPo battery selected for our experiments weighs six kilograms and consists of 12 cells with a total nominal voltage of \SI{44.4}{\V} and a nominal capacity of \SI{22000}{\milli\A\hour}. Thus, it has a nominal energy of $E = $ \SI{976.8}{\W\hour}$ = $ \SI{3516.48}{\kilo\J}. In addition, we set the maximum DoD $\epsilon = 0.80$. Finally, we assume a charge rate of 1\,C, which means that the battery can be completely charged in one hour and $P^{\text{C}} = $ \SI{3516.48}{\kilo\J\per\hour}.

\paragraph{Selection of drone customers}
A customer $j \in C$ can only be supplied by a drone if the mass of the package $m_j$ is lass than the payload of the drone. We assume that 90\% of all packages are below the drone's payload and range from \SI{0.05}{\kg} to \SI{5}{\kg}. The other 10\% range from \SI{5}{\kg} to \SI{50}{\kg}. Hence, with probability $p \in [0, 1)$, we draw $m_j$ from interval $[0.05, 5]$ if $p \leq 0.9$ and from interval $(5, 50)$ otherwise.
However, a package may not be eligible for drone delivery even though its weight is below the payload. This can occur, for example, if a customer is not willing to be supplied by a drone, which may well be the case when a new technology is introduced. In our computational studies, we assume that 75\% of all customers allow drone delivery of their packages.

\paragraph{Distances, travel times, and time parameters}
We use openrouteservice.org \cite{openrouteservice-2019} to obtain actual road network distances and travel times between all locations. The beeline distances for drone flights are determined with GeoPy \cite{geopy-2019}. Both are free-to-use Python packages. The maximum route duration $M$ is eight hours. For each customer $j \in C$, we set the service time of a truck delivery $\tau_j^{\text{S,T}}$ at 120 seconds and the service time of drone delivery $\tau_j^{\text{S,D}}$ at 90 seconds. For depot nodes 0 and $c+1$ the times are fixed at zero. The time needed to prepare the launch of a drone $\tau^{\text{L}}$ is 60 seconds. Unless otherwise noted, the maximum time a truck is allowed to remain stationary at a node  $\tau^{\text{MS}}$ and the maximum time a drone is allowed to hover at retrieval  $\tau_j^{\text{MH}}$ are set high enough that they are not constraining.

\paragraph{Costs} We consider fuel costs of the trucks, wages of the truck drivers, and energy costs of the drones as described in the objective function \eqref{eq:objfunction}. These costs differ between different truck-types, regions, and companies and vary over time. The costs per distance unit traveled by truck $\lambda$ in our experiments are based on a typical P70 UPS truck. We assume a fuel consumption of 11\,mpg (\SI{0.214}{\l\per\km}) for rural areas \cite{lammertm-2012} and a diesel price of \SI{0.76}[\$]{\per\l}. Thus, distance cost parameter $\lambda$ is approximately \SI{0.16}[\$]{\per\km}. Furthermore, we assume that a driver costs approximately $\beta = \,$\SI{20}[\$]{\per\hour} and the electricity rate $\gamma =\,$ \SI[sticky-per]{0.09}[\$]{\per\kWh} (\SI{0.025}[\$]{\per\kilo\J}).

\subsection{Results for small instances}
We use 10 small instances with 20 customers to assess the following:
\begin{compactenum}[1)]
	\item the impact of our preprocessing methods on the runtime,
	\item the influence of varying drone speeds on the costs in the VRPD, and
	\item the benefits of speed selection in comparison to a single fixed speed.
\end{compactenum}  
We have chosen five possible drone speeds ranging from \SI{8}{\m\per\s} to \SI{16}{\m\per\s} in steps of \SI{2}{\m\per\s}. Therefore, we have five VRPDs with $|V| = 1$ and one VRPD-DSS with $V = \{ 8,10,12,14,16 \}$.
Table~\ref{tab:characteristics_20customers} shows the characteristics of each of the 20 customer instances. 
\begin{table}
	\centering
	\footnotesize
	\begin{tabular}{lrrrrrrrrr}
		\toprule
		Instance & $|\bar{C}|$ & \multicolumn{5}{c}{$|W|$ for  $|V| = 1$} & \multicolumn{3}{c}{$|W|$ for $|V| = 5$} \\
		 \cmidrule(lr{0.5em}){3-7}  \cmidrule(lr{0.5em}){8-10}  
		& & 8 & 10 & 12 & 14 & 16 & No OE & OE & $\Delta$W\,[\%] \\
		\midrule
        \texttt{SF\_20\_1}& 13 & 51&80&86&59&29&305&224&-26.56\\
        \texttt{SF\_20\_2}& 15 & 52&81&81&65&37&316&223&-29.43\\
        \texttt{SF\_20\_3}& 11 & 299&397&384&298&198&1576&1249&-20.75\\
        \texttt{SF\_20\_4}& 15 & 78&122&129&99&69&497&409&-17.71\\
        \texttt{SF\_20\_5}& 15 & 176&248&254&207&147&1032&909&-11.92\\
        \texttt{SF\_20\_6}& 12 & 84&134&135&96&54&503&355&-29.42\\
        \texttt{SF\_20\_7}& 14 & 101&173&170&110&61&615&489&-20.49\\
        \texttt{SF\_20\_8}& 12 & 81&131&123&91&46&472&332&-29.66\\
        \texttt{SF\_20\_9}& 13 & 83&131&125&94&70&503&376&-25.25\\
        \texttt{SF\_20\_10}& 14 & 129&208&195&133&74&739&593&-19.76\\
        \midrule
        Avg.& 13.30 & 113.40&170.50&168.20&125.20&78.50&655.80&515.90&-23.10\\
		\bottomrule
	\end{tabular}
	\caption{Characteristics of instances with 20 customers}
	\label{tab:characteristics_20customers}
\end{table}
It includes the number of customers that are available for drone deliveries $|\bar{C}|$ and the number of operations $|W|$ for each VRPD with $|V| = 1$ and for the VRPD-DSS with $|V| = 5$. For the VRPD-DSS, we also show the number of operations without the elimination of dominated drone operations (No OE), with operation elimination (OE), and the percentage of dominated operations that can be eliminated ($\Delta$W).

\subsubsection{Performance improvements through preprocessing}
The average results of all 10 instances for each set of speeds are presented in Table~\ref{tab:results_small_instances}. We solve each instance with two algorithm configurations. First, we apply only the model plus the cuts introduced in Section~\ref{sec:valid_inequalities} (Model + Cuts). The second configuration includes our preprocessing methods (PP + Model + Cuts). We perform five runs per instance to deal with the performance variability in the MILP solution process. The number of nonzero matrix elements ({\#}NZ) following presolve performed by Gurobi is used to represent the size of a problem. In addition, we use the run-time to optimality in seconds (Time) as the performance indicator and show the optimal costs (Costs). Finally, the relative change $\Delta$ between the two configurations for {\#}NZ and Time is displayed as a percentage. 
\begin{table}[ht]
	\centering
	\footnotesize
	\begin{tabular}{rrrrr|rrr|rr}
		\toprule
		$|D|$ & $V$ & \multicolumn{3}{c}{Model + Cuts} & \multicolumn{3}{c}{PP + Model + Cuts} & \multicolumn{2}{c}{$\Delta$[\%]}\\
		\cmidrule(lr{0.5em}){3-5} \cmidrule(lr{0.5em}){6-8} \cmidrule(lr{0.5em}){9-10}
		& & {\#}NZ &  Time\,[s] & Costs\,[\$] & {\#}NZ & Time\,[s] & Costs\,[\$] & {\#}NZ & Time \\
		\midrule
        1&8&45731.68&536.79&107.11&24979.10&174.12&107.11&-45.38&-67.56\\
        &10&47185.30&589.10&104.09&26646.70&144.95&104.09&-43.53&-75.39\\
        &12&49132.62&589.61&102.39&26061.90&169.71&102.39&-46.96&-71.22\\
        &14&43792.16&156.12&103.33&24289.40&39.43&103.33&-44.53&-74.74\\
        &16&35659.88&44.24&104.44&22263.70&11.69&104.44&-37.57&-73.57\\
        \cmidrule{2-10}
        &8,10,12,14,16&65508.24&589.35&101.80&32409.70&211.42&101.80&-50.53&-64.13\\
        \midrule
        2&8&65352.20&315.40&104.10&34540.30&81.62&104.10&-47.15&-74.12\\
        &10&89139.50&595.93&99.96&37114.60&145.78&99.96&-58.36&-75.54\\
        &12&86646.60&276.35&97.70&36017.00&56.02&97.70&-58.43&-79.73\\
        &14&71697.10&91.34&99.07&32967.90&22.92&99.07&-54.02&-74.91\\
        &16&51544.70&28.33&100.74&29663.60&7.76&100.74&-42.45&-72.61\\
        \cmidrule{2-10}
        &8,10,12,14,16&131450.60&698.27&96.97&48416.50&332.00&96.97&-63.17&-52.45\\
		\bottomrule
	\end{tabular}
	\caption{Average results for different drone speeds for instances with 20 customers}
	\label{tab:results_small_instances}
\end{table}

The results show that the number of nonzero elements in the constraint matrix can be reduced significantly by the preprocessing steps. However, optimal solutions are not excluded since costs are the same with and without preprocessing for all instances. The problem size reduction is larger for the VRPD-DSS than for a single-speed VRPD. In the VRPD, only unnecessary variables $z$ can be eliminated, while in the case of the VRPD-DSS, dominated drone operations are also removed. On average, 23.10\,\% of the drone operations are removed by applying our dominance rules, as shown in Table~\ref{tab:characteristics_20customers}.
Table~\ref{tab:results_small_instances} also highlights that a reduced problem size leads to significantly faster run times. However, in contrast to the problem size reductions, the run-time reductions are smaller for the VRPD-DSS than for the single VRPDs. These results demonstrate that our preprocessing methods introduced in Section~\ref{sec:Preprocessing} are highly effective for the considered test instances and therefore, will be used in all further tests.

\subsubsection{Impact of different drone speeds for the VRPD}
Table~\ref{tab:characteristics_20customers} clearly shows that the number of feasible drone operations is heavily dependent on the selected speed. The number of operations first increases with drone speed and, then, it decreases again. Thus, flying faster than a certain threshold reduces the range of a drone due to the nonlinear energy-consumption function (see Figure~\ref{fig:powerAndEnergB}). For the chosen drone model, using a speed of \SI{10}{\metre\per\second} leads to the largest number of operations on average. However, in some instances, a speed of \SI{12}{\metre\per\second} generates the most feasible operations.

Nevertheless, the results in Table~\ref{tab:results_small_instances} demonstrate that more feasible operations do not necessarily lead to lower costs. On average across all 10 instances, the lowest costs considering a single-speed problem can be obtained for both one and two drones with a drone speed of \SI{12}{\metre\per\second}. A faster drone is, therefore, not necessarily advantageous and can result in higher costs. This supports the findings in \cite{rajr-2020}. 
Although using a speed of \SI{12}{\metre\per\second} leads to the lowest costs on average, this does not apply to each individual instance. Figure~\ref{fig:costs_instance_speed} shows the costs for each speed for all 20 customer instances separated by the number of available drones. In some instances, there is almost no difference between two or more speeds (e.g., \texttt{SF\_20\_5}), whereas in other instances, the difference between best and second-best speed is fairly high (e.g., \texttt{SF\_20\_6}). Moreover, the speed that results in the lowest costs for an instance can depend on the number of drones available. For example, with instance \texttt{SF\_20\_4}, \SI{16}{\metre\per\second} leads to minimal cost with one drone, while with two drones, \SI{12}{\metre\per\second} is the best choice. In general, it can be summarily stated that the speed selected in advance for the VRPD can have a significant impact on the costs. To address this issue, drone speed should be included in the decision-making process, such as in the VRPD-DSS. In addition, although performing all flights at speed of \SI{16}{\metre\per\second} leads to substantially fewer operations than a speed of \SI{8}{\metre\per\second}, the average costs are lower. This further illustrates the trade-off between the ability to serve more customers with drones by flying slower, and the savings that can be achieved through shorter delivery times by flying faster.  

\begin{figure}[ht]
	\input{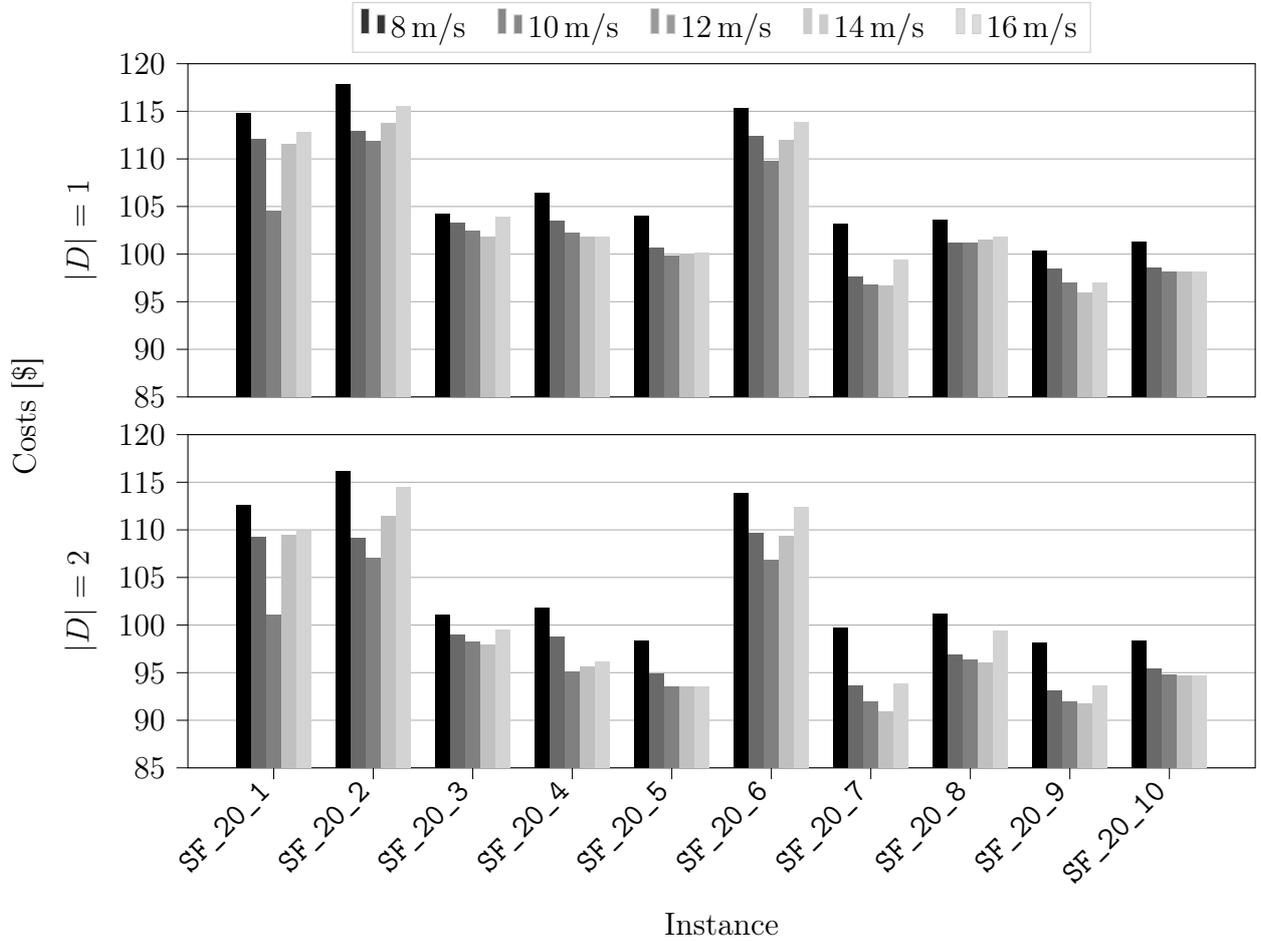}
	\caption{Costs for each 20 customer instance and varying speeds}
	\label{fig:costs_instance_speed}
\end{figure}

\subsubsection{VRPD vs. VRPD-DSS}
The optimal solution of the VRPD-DSS is always at least as good as the best solution of all VRPDs with a single speed. In addition, the costs of the VRPD-DSS are often lower because using different speeds is beneficial in terms of energy consumption or delivery time. However, the cost deviations between the VRPD and VRPD-DSS vary. Table~\ref{tab:SVP_vs_VSP} shows the minimum, average, and maximum percentage costs deviation ($\Delta$Costs) of all 10 instances for each VRPD compared to the VRPD-DSS. In addition, the total number of operations ({\#}OP) at each speed used in all 10 optimal solutions of the VRPD-DSS is given.

\begin{table}[ht]
	\centering
	\footnotesize
	\begin{tabular}{rrrrr|rrrr}
		\toprule
		$V$ & \multicolumn{4}{c}{$|D| = 1$} & \multicolumn{4}{c}{$|D| = 2$} \\
		\cmidrule(lr{0.5em}){2-5} \cmidrule(lr{0.5em}){6-9}
		& \multicolumn{3}{c}{$\Delta$\,Costs\,[\%]} & & \multicolumn{3}{c}{$\Delta$\,Costs\,[\%]} & \\
		\cmidrule(lr{0.5em}){2-4} \cmidrule(lr{0.5em}){6-8}
		& Min &  Avg & Max & {\#}OP &  Min &  Avg & Max & {\#}OP \\
		\midrule
		 8 & 2.46 & 5.20 & 10.60 & 0 & 3.45 & 7.34 & 11.71 & 0\\
		10 & 0.35 & 2.23 & 8.01 & 6 & 1.32 & 3.07 & 8.37 & 10\\
		12 & 0.01 & 0.59 & 1.78 & 22 & 0.00 & 0.78 & 1.89 & 29\\
		14 & 0.04 & 1.46 & 7.45 & 8 & 0.01 & 2.08 & 8.57 & 8 \\
		16 & 0.34 & 2.54 & 8.68 & 0 & 0.04 & 3.79 & 9.00 & 7\\
		\bottomrule
	\end{tabular}
	\caption{Comparison of solutions with a single speed and solutions with speed selection for 20 customers}
	\label{tab:SVP_vs_VSP}
\end{table}

Of course, the average costs of the VRPD with speed \SI{12}{\metre\per\second} deviate the least from the optimal solution of the VRPD-DSS since it has the lowest costs, on average, of all VRPDs. They are, on average, 0.59\,\% (one drone) and 0.78\,\% (two drones) worse than the optimal costs of the VRPD-DSS. For two instances in the case of a single drone and for one instance when two drones are available, the costs are the same, i.e., all flights are performed at \SI{12}{\metre\per\second} although other speeds are available. In contrast, the deviation is over 1\,\% in some instances. All other speeds lead to higher cost deviations on average and in the best and worst cases.

As a result, most flights in the speed selection problem are performed at a speed of \SI{12}{\metre\per\second} for both one and two drones. The slowest speed, \SI{8}{\metre\per\second}, is never used in any solution, and the fastest speed, \SI{16}{\metre\per\second}, is never selected if the tandem has only one drone. We also observe that the deviations between the VRPDs and the VRPD-DSS are larger with two drones. Thus, it is especially important to consider different drone speeds when multiple drones are available.

\subsection{Results for larger instances}
In our further studies, we use larger instances with 30, 40, and 50 customers to gain insights into the benefits of truck-drone tandems under the realistic circumstances presented in this paper. In contrast to the small instances, we limit the maximum time a truck can stop at a node to $\tau^{\text{MS}}$ four minutes, which is twice the service time for a customer visited by a truck. Moreover, the maximum hover time at a retrieval node $\tau^{\text{MH}}$ is restricted to two minutes. Preliminary tests on 20 customer instances have shown that these values improve computational performance compared to the unrestricted case but increase costs only slightly. 

\subsubsection{The MILP solver as a heuristic}
We focus on using the MILP solver as a heuristic rather than as an exact approach in the tests with larger instances. Today, state-of-the-art solvers contain powerful primal heuristics to find good feasible solutions quickly \cite{berthold-2013}. To test Gurobi's ability to provide good solutions quickly, we conduct two different experiments. 

In the first experiment (Experiment 1), we set the Gurobi parameter \textit{MIPFocus} to 1 and \textit{Heuristics} to 0.75. The former modifies the high-level solution strategy to focus on finding good feasible solutions, while the latter lets Gurobi spend even more time on primal heuristics. Using this setting, we perform five runs with different seed values for each instance and limit the maximum time per run to one hour. In the second experiment (Experiment 2), we attempt to achieve a good lower bound. For this purpose, we use the default values of \textit{MIPFocus} and \textit{Heuristics} and provide the best found solution in the first experiment as the starting solution. In addition, we increase the maximum run time to eight hours but perform only a single run per instance. Detailed results for both experiments are shown in Table~\ref{tab:detailed_results_solver} in the appendix.

Table~\ref{tab:results} displays the average over all 10 instances per instance class of: the average objective function value at termination over all five runs of Experiment 1 ($\overline{\text{Obj}}$); the coefficient of variation (CV), i.e., the ratio of the standard deviation to the mean, of the objective function value as a percentage; the objective function value of the best known solution (BKS); the optimality gap of BKS at termination (Gap) as a percentage; and the relative percentage deviation $\overline{\text{RPD}} = \left(\overline{\text{Obj}} - \text{BKS} \right)/\text{BKS} \cdot 100$ at different points in time. Note that BKS corresponds to the objective function value at the termination of Experiment~2.
\begin{table}
    \small
    \centering
    \begin{tabular}{llrrrrrrrrrr}
    \toprule
         $|C|$ & $|D|$ & $\overline{\text{Obj}}$ & CV\,[\%] & BKS & Gap\,[\%] & \multicolumn{6}{c}{$\overline{\text{RPD}}$\,[\%]}\\
                        \cmidrule{7-12}
               &       &    &   &   & & \SI{60}{\second} & \SI{120}{\second} & \SI{300}{\second} & \SI{600}{\second} & \SI{1800}{\second} & \SI{3600}{\second} \\
    \midrule
        30 & 1 & 116.39 & 0.06 & 116.23 & 2.64 & 0.47 & 0.33 & 0.27 & 0.13 & 0.13 & 0.13 \\
           & 2 & 108.81 & 0.27 & 108.44 & 4.09 & 2.39 & 1.49 & 0.93 & 0.71 & 0.47 & 0.33\\
        \midrule
        40 & 1 & 129.23 & 0.25 & 128.95 & 7.43 & 1.75 & 1.20 & 0.85 & 0.61 & 0.28 & 0.22  \\
           & 2 & 119.30 & 0.37 & 118.79 & 12.06 & 6.97 & 2.92 & 1.23 & 0.97 & 0.58 & 0.43 \\
        \midrule
        50 & 1 & 149.01 & 0.50 & 148.29 & 9.30 & 7.18 & 4.18 & 2.07 & 1.55 & 0.76 & 0.49  \\
           & 2 & 138.08 & 1.00 & 136.04 & 15.37 & 21.95 & 15.11 & 5.35 & 3.01 & 1.91 & 1.50 \\           
    \bottomrule
    \end{tabular}
    \caption{Aggregated results for experiments with larger instances}
    \label{tab:results}
\end{table}

The results show that it is difficult to prove optimality with the given approach for VRPD-DSS instances with just 30 customers, and it becomes more difficult as the number of customers and drones increases. However, the solver is able to consistently provide good solutions in a reasonable amount of time. Similar to the optimality gap, the coefficient of variation of the objective function value increases with a growing number of customers and drones. This means that the spread of the objective function values at the end of Experiment~1 increases and the consistency decreases noticeably. Nevertheless, we consider an average coefficient of variation of 1\,\% as a small and acceptable spread. In addition to consistency, we assess Gurobi as providing good-quality solutions for the given instances. For instances with 30 and 40 customers, the average $\overline{\text{RPD}}$ is less than one percent within \SI{600}{\second}. After one hour of run time, the average $\overline{\text{RPD}}$ is between 0.13\,\% and 1.50\,\% for all instance sizes and only greater than 1\,\% for 50 customers and two drones per tandem. Note that the best known solution is almost always identical to the best solution found in Experiment~1 ($\text{RPD}^*$ in Table~\ref{tab:detailed_results_solver} is almost always 0). Thus, the best solution of an instance in Experiment~1 can very rarely be improved in eight hours in Experiment~2.

\subsubsection{Benefits of truck-drone tandems}
Finally, we analyze the benefits and cost-savings of truck-drone tandems for the VRPD-DSS compared to traditional truck-only delivery (TO). We use the best solution for each instance obtained in the experiments to determine potential savings. Detailed information on the results are provided in the appendix. Table~\ref{tab:results_truck_only} presents information on the instances and truck-only delivery, while Table~\ref{tab:detailed_costs_1D} and Table~\ref{tab:detailed_costs_2D} show detailed results for tandems with one and two drones, respectively.

Table~\ref{tab:agg_result_solutions} displays the average solution information for different numbers of customers and drones. It includes the operating time of trucks (Time); the distance traveled by trucks (Dist); the number of operations per drone ({\#}OP); the distance covered per drone (Dist); the number of charge cycles per drone ({\#}CC); the  different cost components, i.e., wages, fuel, and power; the total costs (Total); and finally, the relative change in the total costs compared to truck-only delivery ($\Delta$TO) as percentages.

\begin{table}[]
    \centering
    \scriptsize
    \begin{tabular}{llrrrrrrrrrr}
    \toprule
    $|C|$ & $|D|$ & \multicolumn{2}{c}{Truck} & \multicolumn{3}{c}{Drone} & \multicolumn{4}{c}{Costs\,[\$]} & $\Delta$TO\,[\%] \\
    \cmidrule(lr{0.5em}){3-4}\cmidrule(lr{0.5em}){5-7}\cmidrule(lr{0.5em}){8-11}
    & & Time\,[min] & Dist\,[km] & {\#}OP & Dist\,[km] & {\#}CC & Wages & Fuel & Power & Total\\
    \midrule
    30 & 1 & 270.73&160.76&5.90&37.15&2.92&90.25&25.72&0.26&116.23&-12.56\\
    & 2 & 250.42&153.26&4.90&32.36&2.52&83.48&24.52&0.45&108.44&-18.36 \\
    \midrule
    40 & 1 & 304.29&170.06&8.00&43.68&3.48&101.44&27.21&0.31&128.95&-14.68\\
    & 2 & 277.51&160.96&6.60&37.27&2.96&92.51&25.76&0.52&118.79&-21.41\\
    \midrule
    50 & 1 & 352.78&189.52&10.50&50.03&4.12&117.60&30.33&0.36&148.29&-14.75\\
    & 2 & 321.69&176.18&8.15&43.33&3.48&107.24&28.19&0.61&136.04&-21.89\\
    \bottomrule
    \end{tabular}
    \caption{Aggregated information on solutions for tandems with one and two drones.}
    \label{tab:agg_result_solutions}
\end{table}

The results show that significant cost-savings can be achieved by using truck-drone tandems and that savings increase with an additional drone. However, the benefit of the second drone is less than the benefit of the first drone. As expected, the total number of drone operations and the total distance traveled by all drones increase when two drones are used instead of one drone per tandem. Thus, more customers are served by drones, which results in lower costs. However, the workload per drone decreases, which may result in longer life spans of drones and batteries. Furthermore, savings increase with the number of customers, i.e., with higher customer density, since higher customer density leads to more feasible drone operations (see $|W|$ in Table~\ref{tab:results_truck_only}). Yet the increase from 40 to 50 customers is very small, so perhaps there is a saturation effect that limits the positive impact of customer density on savings, or the heuristic solutions for instances with 50 customers have poorer quality. 

Finally, Figure~\ref{fig:savings} presents a more detailed insight into the cost components and savings. First, we observe that wages account for the largest share of the costs. Fuel costs are less than a quarter of the total operational costs, while power costs are almost negligible. Note that, although power costs have little impact on the total costs, proper consideration of the energy consumption is critical for feasibility, and including drone operations also reduces wages and fuel costs. Moreover, the average reduction in wages is greater than the average reduction in fuel costs. For example, for instances with 30 customers, the use of tandems with a single drone can reduce wages by 13.7\,\%, while fuel costs can be reduced only by 9.5\,\%. Hence, the application of drones can reduce working hours more than the traveled distance of trucks. This highlights expedited delivery through parallelization of services as one of the key benefits of truck-drone tandems. Therefore, it is advisable to include some element of time in the evaluation of truck-drone tandems when comparing them to traditional truck-only delivery.

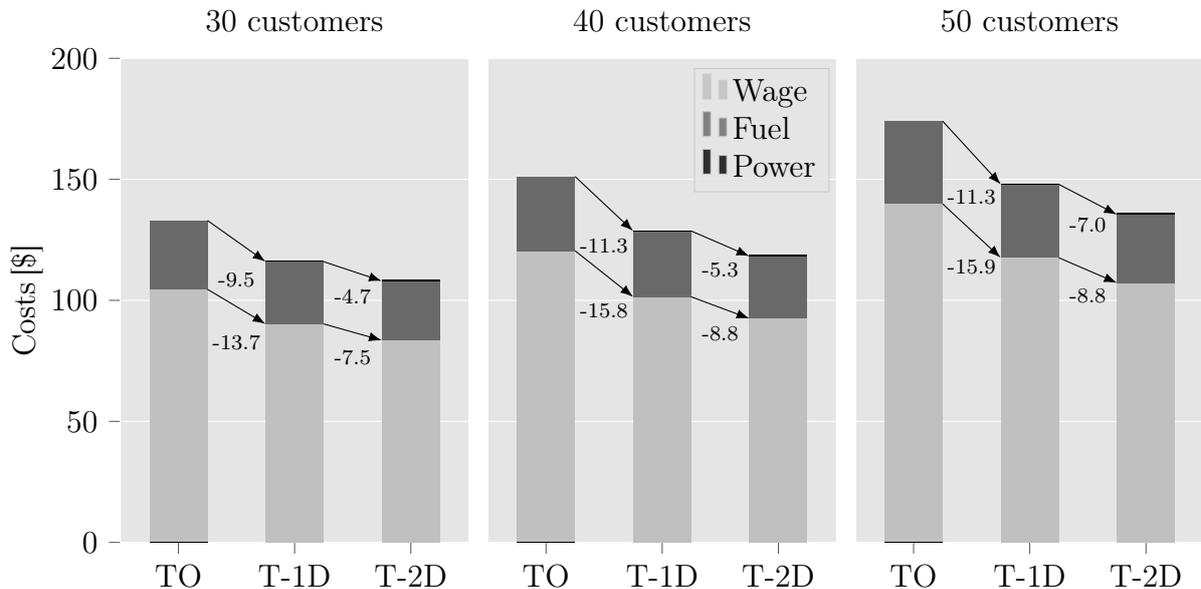
\begin{figure}
	\centering
\begin{tikzpicture}

\begin{groupplot}[width = 0.375\textwidth, height = 8cm, group style={group size=3 by 1, horizontal sep=0.25cm}]
\nextgroupplot[
axis background/.style={fill=white!89.8039215686275!black},
axis line style={white},
legend cell align={left},
legend style={
  fill opacity=0.8,
  draw opacity=1,
  text opacity=1,
  draw=white!80!black,
  fill=white!89.8039215686275!black
},
tick align=outside,
tick pos=left,
title={30 customers},
x grid style={white},
xmin=-0.5, xmax=2.5,
xtick style={color=white!33.3333333333333!black},
xtick={0,1,2},
xticklabels={TO,T-1D,T-2D},
y grid style={white},
ylabel={Costs\,[\$]},
ymajorgrids,
ymin=0, ymax=200,
ytick style={color=white!33.3333333333333!black}
]
\draw[draw=none,fill=white!75.2941176470588!black,very thin] (axis cs:-0.25,0) rectangle (axis cs:0.25,104.57);
\draw[draw=none,fill=white!75.2941176470588!black,very thin] (axis cs:0.75,0) rectangle (axis cs:1.25,90.25);
\draw[draw=none,fill=white!75.2941176470588!black,very thin] (axis cs:1.75,0) rectangle (axis cs:2.25,83.48);

\draw[draw=none,fill=white!41.1764705882353!black,very thin] (axis cs:-0.25,104.57) rectangle (axis cs:0.25,132.99);
\draw[draw=none,fill=white!41.1764705882353!black,very thin] (axis cs:0.75,90.25) rectangle (axis cs:1.25,115.97);
\draw[draw=none,fill=white!41.1764705882353!black,very thin] (axis cs:1.75,83.48) rectangle (axis cs:2.25,108);

\draw[draw=none,fill=black,very thin] (axis cs:-0.25,0) rectangle (axis cs:0.25,0);
\draw[draw=none,fill=black,very thin] (axis cs:0.75,115.97) rectangle (axis cs:1.25,116.23);
\draw[draw=none,fill=black,very thin] (axis cs:1.75,108) rectangle (axis cs:2.25,108.44);

\draw[-Latex] (axis cs:0.25,104.57) -- (axis cs:0.75,90.25) node[midway, below = 7pt]{\scriptsize -13.7}; 
\draw[-Latex] (axis cs:0.25,132.99) -- (axis cs:0.75,115.97) node[midway, below = 8pt]{\scriptsize -9.5}; 

\draw[-Latex] (axis cs:1.25,90.25) -- (axis cs:1.75,83.48) node[midway, below = 3pt]{\scriptsize -7.5}; 
\draw[-Latex] (axis cs:1.25,115.97) -- (axis cs:1.75,108) node[midway, below = 3pt]{\scriptsize -4.7}; 

\nextgroupplot[
axis background/.style={fill=white!89.8039215686275!black},
axis line style={white},
legend cell align={left},
legend style={
  fill opacity=0.8,
  draw opacity=1,
  text opacity=1,
  draw=white!80!black,
  fill=white!89.8039215686275!black
},
tick align=outside,
tick pos=left,
title={40 customers},
x grid style={white},
xmin=-0.5, xmax=2.5,
xtick style={color=white!33.3333333333333!black},
xtick={0,1,2},
xticklabels={TO,T-1D,T-2D},
y grid style={white},
ymajorgrids,
ymin=0, ymax=200,
ytick style={color=white!33.3333333333333!black},
ymajorticks=false
]
\draw[draw=none,fill=white!75.2941176470588!black,very thin] (axis cs:-0.25,0) rectangle (axis cs:0.25,120.48);
\draw[draw=none,fill=white!75.2941176470588!black,very thin] (axis cs:0.75,0) rectangle (axis cs:1.25,101.44);
\draw[draw=none,fill=white!75.2941176470588!black,very thin] (axis cs:1.75,0) rectangle (axis cs:2.25,92.51);
\addlegendimage{ybar,ybar legend,draw=none,fill=white!75.2941176470588!black,very thin}
\addlegendentry{Wage}

\draw[draw=none,fill=white!41.1764705882353!black,very thin] (axis cs:-0.25,120.48) rectangle (axis cs:0.25,151.16);
\draw[draw=none,fill=white!41.1764705882353!black,very thin] (axis cs:0.75,101.44) rectangle (axis cs:1.25,128.65);
\draw[draw=none,fill=white!41.1764705882353!black,very thin] (axis cs:1.75,92.51) rectangle (axis cs:2.25,118.27);
\addlegendimage{ybar,ybar legend,draw=none,fill=white!41.1764705882353!black,very thin}
\addlegendentry{Fuel}

\draw[draw=none,fill=black,very thin] (axis cs:-0.25,0) rectangle (axis cs:0.25,0);
\draw[draw=none,fill=black,very thin] (axis cs:0.75,128.65) rectangle (axis cs:1.25,128.95);
\draw[draw=none,fill=black,very thin] (axis cs:1.75,118.27) rectangle (axis cs:2.25,118.79);
\addlegendimage{ybar,ybar legend,draw=none,fill=black,very thin}
\addlegendentry{Power}

\draw[-Latex] (axis cs:0.25,120.48) -- (axis cs:0.75,101.44) node[midway, below = 7pt]{\scriptsize -15.8}; 
\draw[-Latex] (axis cs:0.25,151.16) -- (axis cs:0.75,128.65) node[midway, below = 9pt]{\scriptsize -11.3}; 

\draw[-Latex] (axis cs:1.25,101.44) -- (axis cs:1.75,92.51) node[midway, below = 3pt]{\scriptsize -8.8}; 
\draw[-Latex] (axis cs:1.25,128.65) -- (axis cs:1.75,118.27) node[midway, below = 3pt]{\scriptsize -5.3}; 

\nextgroupplot[
axis background/.style={fill=white!89.8039215686275!black},
axis line style={white},
legend cell align={left},
legend style={fill opacity=0.5, draw opacity=1, text opacity=1, draw=white!80!black},
tick align=outside,
tick pos=left,
title={50 customers},
x grid style={white},
xmin=-0.5, xmax=2.5,
xtick style={color=white!33.3333333333333!black},
xtick={0,1,2},
xticklabels={TO,T-1D,T-2D},
y grid style={white},
ymajorgrids,
ymin=0, ymax=200,
ytick style={color=white!33.3333333333333!black},
ymajorticks=false
]
\draw[draw=none,fill=white!75.2941176470588!black,very thin] (axis cs:-0.25,0) rectangle (axis cs:0.25,139.85);
\draw[draw=none,fill=white!75.2941176470588!black,very thin] (axis cs:0.75,0) rectangle (axis cs:1.25,117.6);
\draw[draw=none,fill=white!75.2941176470588!black,very thin] (axis cs:1.75,0) rectangle (axis cs:2.25,107.24);

\draw[draw=none,fill=white!41.1764705882353!black,very thin] (axis cs:-0.25,139.85) rectangle (axis cs:0.25,174.03);
\draw[draw=none,fill=white!41.1764705882353!black,very thin] (axis cs:0.75,117.6) rectangle (axis cs:1.25,147.92);
\draw[draw=none,fill=white!41.1764705882353!black,very thin] (axis cs:1.75,107.24) rectangle (axis cs:2.25,135.43);

\draw[draw=none,fill=black,very thin] (axis cs:-0.25,0) rectangle (axis cs:0.25,0);
\draw[draw=none,fill=black,very thin] (axis cs:0.75,147.92) rectangle (axis cs:1.25,148.28);
\draw[draw=none,fill=black,very thin] (axis cs:1.75,135.43) rectangle (axis cs:2.25,136.04);

\draw[-Latex] (axis cs:0.25,139.85) -- (axis cs:0.75,117.60) node[midway, below = 7pt]{\scriptsize -15.9}; 
\draw[-Latex] (axis cs:0.25,174.03) -- (axis cs:0.75,147.92) node[midway, below = 10pt]{\scriptsize -11.3}; 

\draw[-Latex] (axis cs:1.25,117.60) -- (axis cs:1.75,107.24) node[midway, below = 3pt]{\scriptsize -8.8}; 
\draw[-Latex] (axis cs:1.25,147.92) -- (axis cs:1.75,135.43) node[midway, below = 3pt]{\scriptsize -7.0}; 

\end{groupplot}

\end{tikzpicture}
	\caption{Cost structure of different delivery systems and average savings of a tandem with one drone (T-1D) and two drones (T-2D) compared to truck-only delivery (TO)}
	\label{fig:savings}
\end{figure}

\section{Conclusion and future research} \label{sec:Conclusion}
In this paper, combined parcel delivery by trucks and drones is studied, where the speed of a drone flight can be selected from a discrete set of different speeds. We call this problem the vehicle routing problem with drones and drone speed selection, and the following trade-off in speed selection is considered: On one hand, a faster speed shortens delivery times; on the other, it leads to increased energy consumption and, thereby, to a shorter range. We introduce an MILP for this problem, as well as preprocessing methods to eliminate dominated drone speeds and unnecessary variables, and valid inequalities to further strengthen the formulation. We test our approach on instances that closely resemble a real-world scenario in a rural area. The results clearly demonstrate the effectiveness of the preprocessing methods. We also show that, if only a single speed is available for each drone flight, increasing the speed above a threshold does not usually lead to lower costs. However, this threshold differs between instances. Therefore, from a cost perspective, it is always beneficial to consider multiple speeds. The results further indicate that a general solver such as Gurobi can consistently provide high-quality solutions for larger instances of the VRPD-DSS with up to 50 customers. Finally, our results show that truck-drone tandems can achieve significant savings compared to truck-only delivery for the rural scenario considered here. In addition, truck-driver wages account for the largest share of costs but can also be reduced the most by using tandems. In contrast, electricity costs for the drones are almost negligible. However, considering the energy consumption of drones with different speeds is crucial for the feasibility of solutions.

There are many potential avenues for future research. For example, heuristic algorithms can be developed for the VRPD-DSS, and the solutions obtained with the exact approach presented here can be used to evaluate the algorithms. In addition, VRPD-DSS heuristics can be compared to heuristics for the mFSTSP-VDS to investigate the effects of discrete speed levels instead of continuous drone-speed decision variables. Other exact approaches could also be developed to consider continuous drone speeds. A further interesting area of research could be the incorporation of external circumstances such as weather in order to derive more-robust routing decisions.
\bibliography{Literature}
\bibliographystyle{abbrvnat}

\clearpage
\appendix
\setcounter{table}{0}
\section{Tables}

\begin{table}[ht]
    \centering
    \begin{tabular}{lll}
        $n$ & number rotors & 8 \\
        $D$ & diameter of rotor & \SI{0.432}{\m} \\
        $m_{\text{db}}$ & mass drone frame & \SI{10}{\kg} \\
        $m_{\text{b}}$ & mass battery & \SI{6}{\kg} \\
        $c_{\text{db}}$ & drag coefficient drone body & 1.49 \\
        $c_{\text{b}}$ & drag coefficient battery & 1.00 \\
        $c_{\text{p}}$ & drag coefficient package & 2.20 \\
        $A_{\text{db}}$ & projected area drone body & \SI{0.224}{\m\squared} \\
        $A_{\text{b}}$ & projected area battery & \SI{0.015}{\m\squared} \\
        $A_{\text{p}}$ & projected area package & \SI{0.0929}{\m\squared} \\
        $g$ & standard acceleration due to gravity & \SI{9.81}{\m\per\s\squared} \\
        $\rho$ & density of air & \SI{1.2250}{\kg\per\cubic\m}
    \end{tabular}
    \caption{Parameters of the octocopter energy model as in \cite{stolaroffj-2018}}
    \label{tab:parameters_drone}
\end{table}
\begin{landscape}
\begin{table}[ht]
\begin{threeparttable}
    \centering
    \scriptsize
    \begin{tabular}{llrrrrrrrrrrrrrr}
\toprule
    $|C|$ & Instance & \multicolumn{7}{c}{$|D| = 1$} & \multicolumn{7}{c}{$|D| = 2$} \\
    \cmidrule(lr{0.5em}){3-9} \cmidrule(lr{0.5em}){10-16}
& & $\overline{\text{Obj}}$ & $\text{Obj}^*$ & CV & BKS & Gap &$\overline{\text{RPD}}$ & $\text{RPD}^*$ & $\overline{\text{Obj}}$ & $\text{Obj}^*$ & CV & BKS & Gap &$\overline{\text{RPD}}$ & $\text{RPD}^*$\\
\midrule
30&\texttt{SF\_30\_1}&108.51&108.51&0.00&108.51&2.16&0.00&0.00&104.82&104.53&0.35&104.53&5.15&0.28&0.00\\
&\texttt{SF\_30\_2}&122.89&122.33&0.23&122.33&0.00&0.46&0.00&115.91&115.86&0.09&115.01&0.00&0.78&0.74\\
&\texttt{SF\_30\_3}&119.95&119.78&0.17&119.78&4.43&0.14&0.00&112.05&110.90&0.60&110.90&6.72&1.04&0.00\\
&\texttt{SF\_30\_4}&122.51&122.42&0.06&121.95&4.96&0.46&0.39&110.22&110.22&0.00&110.22&5.95&0.00&0.00\\
&\texttt{SF\_30\_5}&117.20&117.20&0.00&117.20&0.00&0.00&0.00&109.42&109.09&0.61&109.09&1.39&0.30&0.00\\
&\texttt{SF\_30\_6}&119.68&119.67&0.00&119.67&1.52&0.01&0.00&111.30&111.15&0.26&111.15&2.48&0.13&0.00\\
&\texttt{SF\_30\_7}&113.36&113.36&0.00&113.36&4.70&0.00&0.00&105.62&105.56&0.11&105.56&6.33&0.06&0.00\\
&\texttt{SF\_30\_8}&104.96&104.96&0.00&104.96&2.71&0.00&0.00&98.58&98.08&0.41&98.08&3.45&0.51&0.00\\
&\texttt{SF\_30\_9}&104.39&104.39&0.00&104.39&2.12&0.00&0.00&98.97&98.96&0.01&98.96&3.86&0.01&0.00\\
&\texttt{SF\_30\_10}&130.42&130.14&0.11&130.14&3.79&0.22&0.00&121.18&120.94&0.24&120.94&5.60&0.20&0.00\\
\midrule
Avg&&116.39&116.28&0.06&116.23&2.64&0.13&0.04&108.81&108.53&0.27&108.44&4.09&0.33&0.07\\
\midrule
40&\texttt{SF\_40\_1}&127.71&127.71&0.00&127.71&7.74&0.00&0.00&119.85&119.15&0.42&119.15&12.32&0.59&0.00\\
&\texttt{SF\_40\_2}&126.90&126.63&0.28&126.63&6.08&0.21&0.00&117.70&117.63&0.09&117.63&8.43&0.06&0.00\\
&\texttt{SF\_40\_3}&132.23&132.23&0.00&132.23&6.74&0.00&0.00&122.76&122.76&0.00&122.76&7.85&0.00&0.00\\
&\texttt{SF\_40\_4}&119.02&118.29&0.38&118.29&7.97&0.62&0.00&109.38&108.83&0.43&108.83&18.51&0.51&0.00\\
&\texttt{SF\_40\_5}&127.12&127.12&0.00&127.12&5.72&0.00&0.00&116.87&116.48&0.28&116.48&11.28&0.33&0.00\\
&\texttt{SF\_40\_6}&131.04&130.44&0.56&130.44&8.33&0.46&0.00&121.57&120.04&0.89&120.04&12.82&1.27&0.00\\
&\texttt{SF\_40\_7}&134.06&133.90&0.24&133.90&8.68&0.12&0.00&125.34&124.70&0.33&124.70&12.21&0.51&0.00\\
&\texttt{SF\_40\_8}&128.36&127.99&0.26&127.99&6.21&0.29&0.00&118.94&118.42&0.38&118.42&13.72&0.44&0.00\\
&\texttt{SF\_40\_9}&134.37&134.20&0.26&134.20&8.92&0.13&0.00&121.97&121.41&0.68&121.41&14.23&0.46&0.00\\
&\texttt{SF\_40\_10}&131.52&131.02&0.55&131.02&7.89&0.38&0.00&118.57&118.44&0.23&118.44&9.22&0.11&0.00\\
\midrule
Avg&&129.23&128.95&0.25&128.95&7.43&0.22&0.00&119.30&118.79&0.37&118.79&12.06&0.43&0.00\\
\midrule
50&\texttt{SF\_50\_1}&151.61&150.63&0.34&150.63&12.91&0.65&0.00&135.10&133.91&1.05&133.91&14.91&0.89&0.00\\
&\texttt{SF\_50\_2}&154.33&153.62&0.38&153.62&9.06&0.46&0.00&143.82&143.11&0.53&143.11&19.36&0.50&0.00\\
&\texttt{SF\_50\_3}&147.26&145.66&1.24&145.66&8.01&1.10&0.00&141.81&134.83&2.74&134.83&15.53&5.18&0.00\\
&\texttt{SF\_50\_4}&150.35&150.18&0.23&150.18&8.76&0.11&0.00&137.67&137.13&0.43&137.13&11.88&0.39&0.00\\
&\texttt{SF\_50\_5}&147.10&146.02&0.97&146.02&8.90&0.74&0.00&136.86&134.99&1.31&134.99&17.54&1.39&0.00\\
&\texttt{SF\_50\_6}&147.32&146.94&0.15&146.94&8.35&0.26&0.00&135.97&134.42&0.77&133.71&13.28&1.69&0.53\\
&\texttt{SF\_50\_7}&162.37&161.46&0.66&161.46&8.50&0.56&0.00&151.34&149.99&0.49&149.99&11.70&0.90&0.00\\
&\texttt{SF\_50\_8}&143.42&142.36&0.81&142.36&7.91&0.74&0.00&131.78&130.77&0.49&130.77&15.03&0.77&0.00\\
&\texttt{SF\_50\_9}&139.75&139.75&0.00&139.75&9.55&0.00&0.00&130.28&127.69&1.37&127.69&16.58&2.03&0.00\\
&\texttt{SF\_50\_10}&146.61&146.28&0.26&146.28&11.02&0.23&0.00&135.91&134.24&0.82&134.24&17.84&1.24&0.00\\
\midrule
Avg&&149.01&148.29&0.50&148.29&9.30&0.49&0.00&138.05&136.11&1.00&136.04&15.37&1.50&0.05\\
\bottomrule
    \end{tabular}
    \begin{tablenotes}
    \item $\overline{\text{Obj}}$ - Average objective function value at termination (Experiment 1)
    \item $\text{Obj}^*$ - Best objective function value at termination (Experiment 1)
    \item CV - Coefficient of variation
    \item BKS - Best known solution, corresponds to the objective function value at termination (Experiment 2)
    \item Gap - Optimality gap of BKS at termination in percent (Experiment 2)
    \item $\overline{\text{RPD}}$ - Relative percentage deviation of $\overline{\text{Obj}}$ with respect to BKS
    \item $\text{RPD}^*$ - Relative percentage deviation of $\text{Obj}^*$ with respect to BKS
    \end{tablenotes}
    \end{threeparttable}
    \caption{Detailed results for MILP solver as heuristic for experiments with larger instances}
    \label{tab:detailed_results_solver}

\end{table}
\end{landscape}

\begin{table}[ht]
	\setlength{\tabcolsep}{3pt}
	\centering
	\scriptsize
	\begin{tabular}{llrrrrrrr}
		\toprule
        $|C|$ & Instance & $|\bar{C}|$ & $|W|$ & \multicolumn{2}{c}{Truck} & \multicolumn{3}{c}{Costs\,[\$]}\\
        \cmidrule(lr{0.5em}){5-6} \cmidrule(lr{0.5em}){7-9}
        &&&& Time\,[min] & Dist\,[km] & Wages & Fuel & Total \\
        \midrule
        30&\texttt{SF\_30\_1}&22&5789&290.03&153.15&96.69&24.50&121.19\\
        &\texttt{SF\_30\_2}&18&1393&330.72&187.26&110.25&29.96&140.21\\
        &\texttt{SF\_30\_3}&19&1466&315.10&178.29&105.04&28.53&133.57\\
        &\texttt{SF\_30\_4}&23&1236&337.82&184.52&112.61&29.52&142.14\\
        &\texttt{SF\_30\_5}&21&1938&327.00&190.07&109.01&30.41&139.42\\
        &\texttt{SF\_30\_6}&18&1588&324.42&183.67&108.15&29.39&137.54\\
        &\texttt{SF\_30\_7}&20&2714&306.67&176.89&102.23&28.30&130.53\\
        &\texttt{SF\_30\_8}&18&3901&281.42&163.23&93.81&26.12&119.93\\
        &\texttt{SF\_30\_9}&20&4023&280.03&158.53&93.35&25.37&118.72\\
        &\texttt{SF\_30\_10}&19&1594&343.77&200.57&114.60&32.09&146.69\\
        \midrule
        Avg&&19.8&2564.2&313.70&177.62&104.57&28.42&132.99\\
         \midrule
        40&\texttt{SF\_40\_1}&26&7012&351.95&193.77&117.33&31.00&148.33\\
        &\texttt{SF\_40\_2}&29&6291&351.08&184.20&117.04&29.47&146.51\\
        &\texttt{SF\_40\_3}&25&6055&359.02&199.14&119.68&31.86&151.54\\
        &\texttt{SF\_40\_4}&28&8452&337.88&171.67&112.64&27.47&140.10\\
        &\texttt{SF\_40\_5}&26&5026&363.95&197.57&121.33&31.61&152.94\\
        &\texttt{SF\_40\_6}&25&4796&375.78&202.95&125.27&32.47&157.74\\
        &\texttt{SF\_40\_7}&22&5642&370.98&190.17&123.67&30.43&154.10\\
        &\texttt{SF\_40\_8}&29&7113&367.85&187.38&122.63&29.98&152.61\\
        &\texttt{SF\_40\_9}&37&7139&377.13&203.96&125.72&32.63&158.35\\
        &\texttt{SF\_40\_10}&27&5144&358.42&186.91&119.48&29.91&149.39\\
         \midrule
        Avg&&27.4&6267.0&361.40&191.77&120.48&30.68&151.16\\
         \midrule
        50&\texttt{SF\_50\_1}&29&9792&407.40&211.05&135.81&33.77&169.58\\
        &\texttt{SF\_50\_2}&41&10169&455.78&230.59&151.94&36.89&188.83\\
        &\texttt{SF\_50\_3}&33&11278&412.35&213.13&137.46&34.10&171.56\\
        &\texttt{SF\_50\_4}&33&10573&416.78&213.74&138.94&34.20&173.14\\
        &\texttt{SF\_50\_5}&34&16652&414.42&213.56&138.15&34.17&172.32\\
        &\texttt{SF\_50\_6}&36&18007&401.65&205.79&133.89&32.93&166.82\\
        &\texttt{SF\_50\_7}&31&7793&453.87&229.16&151.30&36.67&187.97\\
        &\texttt{SF\_50\_8}&29&9707&406.98&212.18&135.67&33.95&169.62\\
        &\texttt{SF\_50\_9}&31&19755&406.72&199.94&135.58&31.99&167.57\\
        &\texttt{SF\_50\_10}&36&15933&419.12&207.15&139.72&33.14&172.86\\
         \midrule
        Avg&&33.3&12965.9&419.51&213.63&139.85&34.18&174.03\\   	
		\bottomrule
	\end{tabular}
	\caption{Information on instances and detailed results for truck usage and costs for truck-only delivery.}
	\label{tab:results_truck_only}
\end{table}
\begin{landscape}
\begin{table}[]
    \centering
    \scriptsize
    \begin{tabular}{llrrrrrrrrrrrr}
    \toprule
    $|C|$ & Instance & \multicolumn{2}{c}{Truck} & \multicolumn{3}{c}{Drone} & \multicolumn{4}{c}{Costs\,[\$]} & \multicolumn{3}{c}{$\Delta$TO\,[\%]} \\
    \cmidrule(lr{0.5em}){3-4}\cmidrule(lr{0.5em}){5-7}\cmidrule(lr{0.5em}){8-11}\cmidrule(lr{0.5em}){12-14}
    & & Time\,[min] & Dist\,[km] & {\#}OP & Dist\,[km] & {\#}CC & Wages & Fuel & Power & Total & Wages & Fuel & Total\\
    \midrule
    30&\texttt{SF\_30\_1}&254.97&145.46&6.00&32.72&2.70&85.00&23.27&0.24&108.51&-12.09&-5.02&-10.46\\
    &\texttt{SF\_30\_2}&285.51&168.06&6.00&37.31&2.97&95.18&26.89&0.26&122.33&-13.67&-10.25&-12.75\\
    &\texttt{SF\_30\_3}&278.23&167.54&5.00&32.58&2.58&92.75&26.81&0.23&119.78&-11.70&-6.03&-10.32\\
    &\texttt{SF\_30\_4}&288.59&159.08&5.00&43.41&3.37&96.20&25.45&0.30&121.95&-14.57&-13.79&-14.20\\
    &\texttt{SF\_30\_5}&274.23&159.48&6.00&36.31&3.04&91.42&25.52&0.27&117.20&-16.14&-16.08&-15.94\\
    &\texttt{SF\_30\_6}&279.59&163.71&6.00&39.93&3.12&93.20&26.19&0.27&119.67&-13.82&-10.89&-12.99\\
    &\texttt{SF\_30\_7}&261.63&161.68&7.00&41.20&3.15&87.22&25.87&0.28&113.36&-14.68&-8.59&-13.15\\
    &\texttt{SF\_30\_8}&243.97&146.39&6.00&29.32&2.32&81.33&23.42&0.20&104.96&-13.30&-10.34&-12.48\\
    &\texttt{SF\_30\_9}&240.87&149.09&6.00&34.41&2.71&80.30&23.86&0.24&104.39&-13.98&-5.95&-12.07\\
    &\texttt{SF\_30\_10}&299.70&187.15&6.00&44.30&3.28&99.91&29.94&0.29&130.14&-12.82&-6.70&-11.28\\
    \midrule
    Avg&&270.73&160.76&5.90&37.15&2.92&90.25&25.72&0.26&116.23&-13.68&-9.36&-12.56\\
    \midrule
    40&\texttt{SF\_40\_1}&301.64&167.78&9.00&42.17&3.54&100.56&26.84&0.31&127.71&-14.29&-13.42&-13.90\\
    &\texttt{SF\_40\_2}&300.20&164.11&8.00&42.09&3.46&100.07&26.26&0.30&126.63&-14.50&-10.89&-13.57\\
    &\texttt{SF\_40\_3}&310.97&176.87&7.00&39.47&2.97&103.67&28.30&0.26&132.23&-13.38&-11.17&-12.74\\
    &\texttt{SF\_40\_4}&280.11&153.82&8.00&39.60&3.40&93.38&24.61&0.30&118.29&-17.10&-10.41&-15.57\\
    &\texttt{SF\_40\_5}&298.85&170.02&8.00&41.18&3.33&99.62&27.20&0.29&127.12&-17.89&-13.95&-16.88\\
    &\texttt{SF\_40\_6}&307.60&172.33&7.00&46.97&3.67&102.54&27.57&0.32&130.44&-18.14&-15.09&-17.31\\
    &\texttt{SF\_40\_7}&316.63&175.07&7.00&49.53&3.83&105.55&28.01&0.34&133.90&-14.65&-7.95&-13.11\\
    &\texttt{SF\_40\_8}&301.53&169.67&9.00&46.79&3.64&100.52&27.15&0.32&127.99&-18.03&-9.44&-16.13\\
    &\texttt{SF\_40\_9}&315.48&179.51&8.00&44.42&3.49&105.17&28.72&0.31&134.20&-16.35&-11.98&-15.25\\
    &\texttt{SF\_40\_10}&309.85&171.45&9.00&44.56&3.43&103.29&27.43&0.30&131.02&-13.55&-8.29&-12.30\\
    \midrule
    Avg&&304.29&170.06&8.00&43.68&3.48&101.44&27.21&0.31&128.95&-15.79&-11.26&-14.68\\
    \midrule
    50&\texttt{SF\_50\_1}&356.03&197.39&10.00&49.44&4.15&118.68&31.58&0.36&150.63&-12.61&-6.49&-11.17\\
    &\texttt{SF\_50\_2}&367.72&191.66&9.00&53.40&4.18&122.58&30.67&0.37&153.62&-19.32&-16.86&-18.65\\
    &\texttt{SF\_50\_3}&348.02&182.93&11.00&52.00&4.24&116.01&29.27&0.37&145.66&-15.60&-14.16&-15.10\\
    &\texttt{SF\_50\_4}&356.84&192.77&11.00&53.63&4.31&118.95&30.84&0.38&150.18&-14.39&-9.82&-13.26\\
    &\texttt{SF\_50\_5}&348.54&184.23&10.00&49.66&3.99&116.19&29.48&0.35&146.02&-15.90&-13.73&-15.26\\
    &\texttt{SF\_50\_6}&349.67&187.73&10.00&46.47&3.86&116.56&30.04&0.34&146.94&-12.94&-8.78&-11.92\\
    &\texttt{SF\_50\_7}&383.20&208.28&11.00&53.96&4.41&127.74&33.33&0.39&161.46&-15.57&-9.11&-14.10\\
    &\texttt{SF\_50\_8}&336.55&186.18&10.00&49.22&4.30&112.19&29.79&0.38&142.36&-17.31&-12.25&-16.07\\
    &\texttt{SF\_50\_9}&333.30&176.90&11.00&45.44&3.82&111.11&28.30&0.34&139.75&-18.05&-11.53&-16.60\\
    &\texttt{SF\_50\_10}&347.95&187.16&12.00&47.09&3.91&115.99&29.95&0.34&146.28&-16.98&-9.63&-15.38\\
    \midrule
    Avg&&352.78&189.52&10.50&50.03&4.12&117.60&30.33&0.36&148.29&-15.87&-11.24&-14.75\\
    \bottomrule
    \end{tabular}
    \caption{Detailed information on truck and drone usage and costs for tandems with one drone.}
    \label{tab:detailed_costs_1D}
\end{table}

\begin{table}[]
    \centering
    \scriptsize
    \begin{tabular}{llrrrrrrrrrrrr}
    \toprule
    $|C|$ & Instance & \multicolumn{2}{c}{Truck} & \multicolumn{3}{c}{Drone} & \multicolumn{4}{c}{Costs\,[\$]} & \multicolumn{3}{c}{$\Delta$TO\,[\%]} \\
    \cmidrule(lr{0.5em}){3-4}\cmidrule(lr{0.5em}){5-7}\cmidrule(lr{0.5em}){8-11}\cmidrule(lr{0.5em}){12-14}
    & & Time\,[min] & Dist\,[km] & {\#}OP & Dist\,[km] & {\#}CC & Wages & Fuel & Power & Total & Wages & Fuel & Total\\
    \midrule
    30&\texttt{SF\_30\_1}&245.28&139.65&5.00&29.18&2.36&81.77&22.34&0.42&104.53&-15.43&-8.82&-13.75\\
    &\texttt{SF\_30\_2}&266.55&160.80&4.50&31.24&2.42&88.86&25.73&0.43&115.01&-19.40&-14.12&-17.97\\
    &\texttt{SF\_30\_3}&259.20&150.09&4.50&35.85&2.73&86.41&24.01&0.48&110.90&-17.74&-15.84&-16.97\\
    &\texttt{SF\_30\_4}&256.93&150.56&4.50&32.89&2.73&85.65&24.09&0.48&110.22&-23.94&-18.39&-22.46\\
    &\texttt{SF\_30\_5}&251.92&154.17&4.50&32.37&2.52&83.98&24.67&0.44&109.09&-22.96&-18.88&-21.75\\
    &\texttt{SF\_30\_6}&254.83&160.57&5.50&37.92&2.89&84.95&25.69&0.51&111.15&-21.45&-12.59&-19.19\\
    &\texttt{SF\_30\_7}&240.90&155.17&5.00&32.22&2.44&80.31&24.83&0.43&105.56&-21.44&-12.26&-19.13\\
    &\texttt{SF\_30\_8}&224.10&143.60&5.50&28.71&2.25&74.71&22.98&0.40&98.08&-20.36&-12.02&-18.22\\
    &\texttt{SF\_30\_9}&226.10&144.94&5.50&28.53&2.25&75.37&23.19&0.40&98.96&-19.26&-8.59&-16.64\\
    &\texttt{SF\_30\_10}&278.35&173.06&4.50&34.70&2.62&92.79&27.69&0.46&120.94&-19.03&-13.71&-17.55\\
    \midrule
    Avg&&250.42&153.26&4.90&32.36&2.52&83.48&24.52&0.45&108.44&-20.10&-13.52&-18.36\\
    \midrule
    40&\texttt{SF\_40\_1}&279.58&158.84&6.50&38.19&3.04&93.20&25.42&0.53&119.15&-20.57&-18.00&-19.67\\
    &\texttt{SF\_40\_2}&273.50&161.93&8.00&37.25&3.11&91.18&25.91&0.55&117.63&-22.10&-12.08&-19.71\\
    &\texttt{SF\_40\_3}&285.22&170.34&6.00&29.94&2.40&95.08&27.25&0.42&122.76&-20.55&-14.47&-18.99\\
    &\texttt{SF\_40\_4}&257.80&140.10&6.00&34.85&2.70&85.94&22.42&0.47&108.83&-23.70&-18.38&-22.32\\
    &\texttt{SF\_40\_5}&269.82&162.64&6.50&36.70&2.91&89.95&26.02&0.51&116.48&-25.86&-17.68&-23.84\\
    &\texttt{SF\_40\_6}&279.72&164.11&6.00&39.21&3.06&93.25&26.26&0.54&120.04&-25.56&-19.13&-23.90\\
    &\texttt{SF\_40\_7}&292.78&165.86&6.50&39.65&3.17&97.60&26.54&0.56&124.70&-21.08&-12.78&-19.08\\
    &\texttt{SF\_40\_8}&276.57&160.56&7.00&39.34&3.05&92.20&25.69&0.54&118.42&-24.92&-14.51&-22.53\\
    &\texttt{SF\_40\_9}&282.51&166.69&7.00&39.64&3.20&94.18&26.67&0.56&121.41&-25.09&-18.27&-23.33\\
    &\texttt{SF\_40\_10}&277.64&158.55&6.50&37.88&2.93&92.56&25.37&0.52&118.44&-22.53&-15.18&-20.72\\
    \midrule
    Avg&&277.51&160.96&6.60&37.27&2.96&92.51&25.76&0.52&118.79&-23.20&-16.05&-21.41\\
    \midrule
    50&\texttt{SF\_50\_1}&317.33&172.32&6.50&41.43&3.16&105.78&27.57&0.55&133.91&-22.11&-18.36&-21.03\\
    &\texttt{SF\_50\_2}&338.44&185.44&8.00&45.45&3.53&112.82&29.67&0.62&143.11&-24.98&-20.20&-23.70\\
    &\texttt{SF\_50\_3}&320.17&171.51&9.00&46.13&3.71&106.73&27.44&0.65&134.83&-23.74&-19.93&-22.61\\
    &\texttt{SF\_50\_3}&324.07&177.92&7.50&44.85&3.60&108.03&28.47&0.63&137.13&-22.25&-16.75&-20.80\\
    &\texttt{SF\_50\_5}&319.27&174.94&8.00&39.97&3.23&106.43&27.99&0.57&134.99&-22.96&-18.09&-21.66\\
    &\texttt{SF\_50\_6}&313.15&179.69&8.50&40.57&3.26&104.39&28.75&0.57&133.71&-22.03&-12.69&-19.85\\
    &\texttt{SF\_50\_7}&352.63&198.73&8.50&45.71&3.66&117.55&31.80&0.64&149.99&-22.31&-13.28&-20.21\\
    &\texttt{SF\_50\_8}&307.89&171.80&8.50&44.53&3.66&102.64&27.49&0.64&130.77&-24.35&-19.03&-22.90\\
    &\texttt{SF\_50\_9}&304.77&159.27&8.00&42.20&3.47&101.60&25.48&0.61&127.69&-25.06&-20.35&-23.80\\
    &\texttt{SF\_50\_10}&319.18&170.19&9.00&42.41&3.49&106.40&27.23&0.61&134.24&-23.85&-17.83&-22.34\\
    \midrule
    Avg&&321.69&176.18&8.15&43.33&3.48&107.24&28.19&0.61&136.04&-23.36&-17.65&-21.89\\
    \bottomrule
    \end{tabular}
    \caption{Detailed information on truck and drone usage and costs for tandems with two drones.}
    \label{tab:detailed_costs_2D}
\end{table}
\end{landscape}
\end{document}